\documentclass[12pt]{article}
\usepackage{amsmath, amssymb, amsthm}
\usepackage[svgnames]{xcolor}

\usepackage{amsfonts}
\usepackage{bm}
\usepackage{graphicx, subfigure}
\usepackage{makeidx}
\usepackage{hyperref}
\usepackage{mathtools}

\theoremstyle{plain} 
\newtheorem{thm}{Theorem}[section]

\theoremstyle{definition} 

\theoremstyle{remark} 
\newtheorem{oss}{Remark} 

\usepackage{fancyhdr}

\usepackage{epstopdf}
\usepackage{graphicx, subfigure}

\begin{document}

\title{Viral epidemiology of the adult \textit{Apis Mellifera} infested by the
\textit{Varroa destructor} mite
}

\author{Sara Bernardi, Ezio Venturino\thanks{
Corresponding author's email: ezio.venturino@unito.it}\\
Dipartimento di Matematica ``Giuseppe Peano'',\\
Universit\`{a} di Torino, Italy.
}
\maketitle

\begin{abstract}

The ectoparasitic mite {\textit{Varroa destructor}} has become one of the major worldwide
threats for apiculture.

{\textit{Varroa destructor}} attacks the honey bee {\textit{Apis mellifera}} weakening its
host by sucking hemolymph. However, the damage to bee colonies is not strictly related to
the parasitic action of the mite but it derives, above all, from its action as vector
increasing the trasmission
of many viral diseases
such as acute paralysis (ABPV)
and deformed wing viruses (DWV), that are considered among the main causes of CCD
(Colony Collapse Disorder).

In this work we discuss an SI model that describes how the presence of the mite affects
the epidemiology of these viruses on adult bees.
We characterize the system behavior, establishing that ultimately either only healthy
bees survive, or the disease becomes endemic and mites are wiped out. 
Another dangerous alternative is the \textit{Varroa} invasion scenario with the extinction
of healthy bees. The final possible configuration is
the coexistence equilibrium in which honey bees share their infected hive with mites. 

The analysis is in line with some observed facts in natural honey bee colonies. Namely,
these diseases are endemic. Further, if the mite population is present, necessarily
the viral infection occurs.

The findings of this study indicate that a low horizontal transmission rate of the virus
among honey bees in beehives will help in protecting bee colonies from
\textit{Varroa} infestation and viral epidemics.
\end{abstract}

\section{Introduction}

Pathogens, e.g. viruses, fungi and bacteria, significantly affect the dynamics of invertebrates.
Experimental evidence reveals that some naturally harmless honey bee viruses turn
instead epidemic via new transmission routes \cite{SM}.
Specifically, when the parasitic mite \textit{Varroa jacobsoni}, that normally
affects the eastern honey bee \textit{Apis cerana}, switched host and
started to attack the western honey bee \textit{Apis mellifera},
causing larger outbreaks of mites, in particular of
\textit{Varroa destructor}, in \textit{A. mellifera} than in \textit{A. cerana}. 
These mite populations act as bee viruses vectors, that are presumed to be the cause
of the worldwide loss of mite-infested honey bee colonies, because they induce
a much faster spread of the viruses than before.

In order to better understand the \textit{Varroa}-infested colony response to these
viruses, it is necessary to study the effects of  mites and viral diseases together,
since both pathogens appear simultaneously in fields conditions.  

The main goal of this paper is to formulate and investigate
a mathematical model able to describe how the mite presence affects
the epidemiology of these viruses on adult bees. 
To this end, we take into consideration all the possible transmission routes
for honey bee viral diseases, including the vectorial one through parasitizing mites.

The paper is organized as follows. 
In the next Section we present the biological background,
briefly describing the host \textit{Apis Mellifera},
the vector \textit{Varroa destructor}, the honey bee viral diseases
and the triangular relationship between them. 

In Section \ref{sec:model} the model is formulated, assessing some basic properties,
like well-posedness. Section \ref{sec:equil} investigates its equilibria for
feasibility and the following Section \ref{sec:stab} for the stability analysis.
In Section \ref{sec:sim} we provide a set of reasonable values for the parameters,
those that can be obtained from field data, and perform numerical simulations interpreting
their results. Bifurcations are investigated in Section \ref{sec:bif}, connecting
all the equilibria together and further showing that they completely capture the system's
behavior, as persistent oscillations among the populations are proven to be impossible.
Sensitivity analysis is carried out in the following Section \ref{sec:sens}.
A final discussion concludes the paper.

\section{Biological background}

\subsection{The host: the honey bee \textit{Apis mellifera}}

In the summer season a honey bee colony harbors generally a single reproductive queen,
living on average  $2$ to $3$ years,
with about $60000$ adult female worker bees,
$10000$ to $30000$ individuals at the
brood stage (egg, larvae and pupae) and a few hundreds of male drones, \cite{Ratti}.
This large population of workers takes care of the colony by foraging, producing
honey, caring for and rearing the brood and the bees new generation.
The queen lays fertilized eggs producing worker bees and, rather seldom, other queens,
and non-fertilized eggs from which drones are born.
During the winter the queen and around $8000$ to $15000$ adult workers
thrive feeding only on the honey produced and stored in the previous summer.
Because of the population density experienced in such an environment, the honey
bees can very easily be affected by pathogens, \cite{parametri}.

Honey bees have many defences against diseases: for instance, grooming the mites
away from their bodies, or removing the infested brood.
However, colonies still suffer from a number of diseases and pests.

\subsection{The vector: the mite \textit{Varroa destructor}}

A mature \textit{Varroa} mite is in one of its two life stages,
the phoretic one, when it thrives attacking a bee and sucking its hemolymph.
To this end, it pierces the cuticle of the bee with its specialized mouth-parts.
As the mite can switch host during this phase, it may spread viruses in the colony
by sucking them from one host and then injecting them into a healthy individual, \cite{api}.

The honey bee mortality induced by the subtraction of hemolymph and the tearing
of tissues in the act of sucking is very insignificant.
Therefore, the damage to bee colonies derives from the parasitic action of the mite but,
above all, from its action as vector of many viral seriously harmful diseases.

When grooming, the bees bite the mites and the latter may fall off to the bottom of
the hive, constituting the ``natural mite drop'', which however represents
only a small part of the whole mite population.

In fact, the mites second life stage is the reproductive one, in which they are
hidden under the cell cappings, \cite{api}.
To lay eggs, an adult female mite moves from an adult bee into the cell of a larva.
When the brood cell is capped, the larva begins pupating and the mite starts to feed.
About three days later the mite lays its eggs: one is unfertilized and produces
a male, four to six are fertilized and give rise to females.
After hatching, the females feed on the pupa and mate with the male.
The surviving mature female mites remain attached
to the host bee when the latter becomes adult and leaves the cell.

\subsection{The viral pathogens}

Most of honey bee viruses commonly cause covert infections, i.e. the virus can
be detected at low titers within the honey bee population in the absence of obvious
symptoms in infected individuals or colonies.
However, when injected into the open circulatory system of the insects,
in which the hearth pumps blood into the space in between the organs rather than in
closed vessels as in vertebrates,
these diseases are extremely virulent: just a few viral particles per bee are sufficient
to cause the death within a few days.

The most serious problem caused by \textit{Varroa destructor} is therefore the transmission
of viral diseases to honey bees.
Scientists have found viruses in honey bees since decades, but in general they were
considered harmless. Only when about thirty years ago \textit{Varroa}
became a widespread problem, about twenty different viruses have been
associated with \textit{Varroa} mites, that are their physical or biological vectors,
\cite{api}.
Indeed,
when a virus-carrying mite during its phoretic phase attaches to a healthy bee,
it can transmit the virus to the bee, thereby infecting it.
Further, viral diseases are also transmitted among bees through food,
feces, from queen to egg, and from drone to queen.

On the other hand, a virus-free phoretic mite can become a physical virus carrier
horizontally from other infected mites, but also
when it attaches to an infected bee.
Therefore to control the viruses, the \textit{Varroa} population in a hive should be
kept in check, while symptoms of viral diseases usually indicate the presence of
\textit{Varroa} in the colony, \cite{api}. 


Epidemiological surveys and laboratory experiments have demonstrated that especially
two viruses can successfully be transmitted between
honey bees during mite feeding activities:
deformed wing virus (DWV) and acute paralysis virus (ABPV). We briefly describe
them in the next Subsections.


\subsubsection*{Deformed wing virus (DWV)}

This syndrome in infected bees is rather benign, causing
covert infections with no visible symptoms, \cite{GA}.
But in the presence of \textit{Varroa destructor}, overt DWV infections are much more
common. The affected bees show deformed wings,
bloated and shortened abdomens, discoloration and are bound to die within three days from
the infection. The outcome could even be the whole colony destruction.
Further, DWV appears to replicate in \textit{Varroa}, making it also a biological
vector.



\subsubsection*{Acute paralysis virus (ABPV)}

ABPV was present at low concentrations as a covert
infection in adult bees, but never caused paralysis outbreaks.
However, when \textit{Varroa destructor} established in the European
honey bee populations, the virus has been found in the bees worldwide,
and is now geographically distributed as the \textit{A. mellifera}, \cite{GA}.

Affected bees by this virus are unable to fly, lose the hair from their bodies and
tremble uncontrollably. The virus has been suggested to be a primary cause of bee
mortality. Infected pupae and adults suffer rapid death. 

\section{The model}\label{sec:model}

Let $B$ denote the healthy bees, $I$ the infected ones, $M$ the healthy mites,
$N$ the infected ones.
The model reads as follows:
\begin{eqnarray}\label{model0}
B'&=&b \frac{B}{B+I}-\lambda BN-\gamma BI-m B-aBM-qBN,
\\ \nonumber
I'&=&b \frac{I}{B+I} +\lambda BN+\gamma BI -(m + \mu) I-cIM-dIN ,
\\ \nonumber
M'&=&\tilde R(M+N,B+I)-\beta MI-\delta MN-eMB,
\\ \nonumber
N'&=&-n N- pN(N+M) + \beta MI+\delta MN-eNB.
\end{eqnarray}
The first equation expresses the dynamics of healthy bees.
The parameter $b$ represents the daily birth rate of the bees,
which we assume that are born healthy in proportion of
the fraction of healthy bees in the colony.
This takes into account the fact that
the disease is transmitted to larvae mainly by infected nurse bees
that contamine royal jelly.
Furthermore, the second term models the situation in which, if attacked by infected mites,
the healthy bees may contract the virus carried by the mite at transmission rate $\lambda$.


Horizontal transmission among the bees occurs at rate $\gamma$ either via small wounds of
the exoskeleton, e.g. as a result of hair loss, or ingestion of faeces.
The fourth term in the equation describes the bees' natural mortality, at rate $m$,
while the last ones model the competition due to parasitism.

The second equation describes the evolution of infected bees.
As for the healthy bees, they reproduce at rate $b$,
but in proportion now to the fraction of infected bees in the population.
The second and third terms account for the new recruits is this class due to
horizontal contacts, respectively with infected mites and bees.
Infected bees are subject to both their natural $m$ and disease-related $\mu$ mortalities.
The remaining terms describe again competition.

The third and fourth equations model respectively the dynamics of healthy and infected mites.

The growth of sound mites is described by the $\tilde R$ function.
When attacking the infected bees they can, in turn, contract the virus
with transmission rate $\beta$.
Furthermore, healthy mites can acquire the virus also horizontally from other
infected mites at rate $\delta$. Note that
viral infection is not transmitted to their offsprings;
we therefore exclude vertical transmission among mites. 
The last term in both equations models the grooming behavior of healthy bees
at rate $e$. We exclude this behavior for
infected bees, as they are instead weakened by the virus to be able
to perform any resistance against \textit{Varroa} mites.

Note that in formulating the fourth equation we account for the fact that
the virus does not cause any harm to the infected \textit{Varroa} population,
but individuals remain infected until their natural death, that occurs at rate $n$,
hence there cannot be any disease recovery.
On the other hand, we consider the effect of intraspecific competition between
healthy and infected mites at rate $p$. The remaining terms once again describe
the new recruits in this class, from horizontal contacts with infected bees or mites.

Finally, we make the following simplifying assumptions:
\begin{itemize}
\item the daily bees birth rate $b$ is constant,
\item the honey bee mortality induced by the presence of \textit{Varroa} mite
can be neglected when compared with their natural and disease-related mortality,
so that we set $a=q=c=d=0$; 
\item healthy mites grow in a logistic fashion with reproduction rate $r$
and carrying capacity $K$, so that
$$
\tilde R(M+N,B+I)=r \left(1-\frac{M+N}{K}\right)(M+N).
$$
\end{itemize}
Therefore, system (\ref{model0}) reduces to the following one,
written in shorthand as
\begin{equation}\label{sistbound}
X'=f(X),
\end{equation}
where $X=(B,I,M,N)$ is the population vector and
the components of the right hand side $f$ are explicitly given by:
\begin{eqnarray}\label{model}
B'&=&b \frac{B}{B+I}-\lambda BN-\gamma BI-m B
\\ \nonumber 
I'&=&b \frac{I}{B+I} +\lambda BN+\gamma BI-(m+\mu)I 
\\ \nonumber 
M'&=&r \left(1-\frac{M+N}{K}\right)(M+N)-\beta MI-\delta MN-eMB
\\ \nonumber 
N'&=&-n N-p N (N+M)+\beta MI+\delta MN-eNB
\end{eqnarray}


\subsection{Well posedness and boundedness}

In this section we show that the system \eqref{model} is well-posed.
Basically, we follow the path of \cite{FCB}.

We need this because
the right-hand side of \eqref{model} is not well-defined at the points where $B+I=0$. 
But the next result shows that the solutions are always bounded away from this set.
Let us first define the subset $\mathcal{D}^0$
of $\mathbb{R}^4_+$ of the points that are away from the
singularity, namely
$$
\mathcal{D}^0=\left\{ X=(B,I,M,N) \in \mathbb{R}^4_+ : B+I \neq 0\right\}.
$$

\begin{thm}[Well-posedness and boundedness]
\label{teobound}
If $X_0 \in \mathcal{D}^0$, then there exists a unique solution of \eqref{sistbound}
defined on $[0,+\infty )$ such that $X(0)=X_0$. Moreover, for any $t>0$, $X(t) \in \mathcal{D}^0$,
and 
\begin{equation}
\label{bound1}
\frac{b}{\widetilde{m}} \leq \liminf\limits_{t\rightarrow +\infty} (B(t)+I(t))
\leq \limsup\limits_{t\rightarrow +\infty} (B(t)+I(t)) \leq \frac{b}{m} \\
\end{equation}
where $\widetilde{m} = m+\mu$. Further,
\begin{equation}
M(t)+N(t) \leq L , \quad \forall t \geq 0,
\label{bound2}
\end{equation}
where $L = \max \left\lbrace M(0)+N(0), K \right\rbrace$.
\end{thm}

\begin{proof}
The right-hand side of the system is globally Lipschitz continuous
on $\mathcal{D}^0$, so that existence and uniqueness of the solution of
system \eqref{sistbound} is ensured for every trajectory that is at a finite
distance of this boundary. 
We now show that (\ref{bound1}) and (\ref{bound2}) hold for all the
trajectories originating from a point where $B+I \neq 0$.
From the boundedness
of the variables, it will follow that all trajectories are defined on an infinite horizon
and they stay away from the set
$\left\lbrace (B,I,M,N)  \in \mathbb{R}^4_+ : B+I=0 \right\rbrace $. 

Summing the first two equations in \eqref{model}, for any point inside $\mathcal{D}^0$
we obtain
\begin{equation*}
B'+I'=b-mB-(m+\mu)I \geq b-\widetilde{m}(B+I).
\end{equation*}
Integrating this differential inequality between any two points $X(0)=X_0$ and $X(t)$
of a trajectory for which $X(\tau) \in \mathcal{D}^0$, $\tau \in [0,t]$, we get
\begin{equation}\label{bound1a}
B(t)+I(t) \geq \frac{b}{\widetilde{m}}(1-e^{-\widetilde{m}t})+(B(0)+I(0))e^{-\widetilde{m}t},
\end{equation} 
where the right-hand side is positive for any $t>0$.

Similarly, we have
\begin{equation*}
B'+I' \leq b-m(B+I),
\end{equation*}
and therefore
\begin{equation}
\label{bound1b}
B(t)+I(t) \leq \frac{b}{m}(1-e^{-mt})+(B(0)+I(0))e^{-mt}.
\end{equation}
From \eqref{bound1a} and \eqref{bound1b}, it is easy to see that the inequalities
in \eqref{bound1} hold for any portion of trajectory remaining inside $\mathcal{D}^0$.

We now consider the evolution of $M$ and $N$.
Defining the whole \textit{Varroa} mite population as $V = M+N$ and summing the
last two equations in \eqref{model} yield, for any point inside $\mathcal{D}^0$,
\begin{equation}
V'=r V -r\frac{V^2}{K}-eBV-nN-pNV \leq rV \left(1-\frac{V}{K} \right).
\end{equation}
The latter is a logistic term, thus the solution trajectory moves toward the stable
equilibrium point $V=K$. Hence, we get
$V(t) \leq \max \left\lbrace V(0), K \right\rbrace$ for all $t \geq 0$, proving \eqref{bound2}.

From the results on $B+I$ and $V$ the boundedness of all populations follows. 
\cite{FCB}
Thus all trajectories originatin in $\mathcal{D}^0$
remain in $\mathcal{D}^0$ for all $t>0$.
\end{proof}

Theorem \eqref{teobound} shows that the compact set $\mathcal{D}^1$,
defined as the largest subset of $\mathcal{D}^0$ satisfying
the inequalities of Theorem \eqref{teobound},
\begin{equation*}
\mathcal{D}^1 \doteq \left\lbrace (B,I,M,N) \in \mathbb{R}^4_+ : \frac{b}{\widetilde{m}}
\leq B+I \leq \frac{b}{m}, 0 \leq M+N \leq L \right\rbrace.
\end{equation*}
is a positively invariant attractor for all the system's trajectories,
indicating that the equilibria analysis shows the ultimate behavior of the system.

\section{Equilibria}\label{sec:equil}

Since the domain of definition of \eqref{model} is $\mathcal{D}^0$
we exclude solutions not included in it.

The admissible equilibria are then
$$
E_1=\left(\frac{b}{m}, 0, 0, 0\right), \quad
E_2=\left(0, \frac{b}{m+\mu}, 0, 0\right),
$$
which are always feasible,
$$
E_3=\left(\frac{b}{m}, 0, \frac{K}{r}\left(r-\frac{eb}{m}\right), 0\right)
$$
which is feasible for
\begin{equation}\label{E3feas}
r \ge \frac{eb}{m}.
\end{equation}
Then we find
\begin{equation}\label{E4}
E_4=\left(\frac{\mu^2 -b \gamma+m \mu}{\mu \gamma}, \frac{b \gamma - m \mu}{\mu \gamma},
0, 0\right),
\end{equation}
so that feasibility for $E_4$ is given by
\begin{equation}\label{E4feas}
0<b \gamma - m \mu <\mu^2.
\end{equation}

To find the equilibrium point with no healthy bees, we proceed as follows.
From the second equation of \eqref{model}, we find 
\begin{equation*}
I=\frac{b}{m+\mu}.
\end{equation*}
Substituting the value of $I$ and rearranging, the last two equations 
provide two conic sections: $\psi_1 (M,N)=0$ which explicitly is
\begin{equation}\label{psi1}
-\frac{r}{K}M^2-\left(\frac{2r}{K}+\delta \right)MN
-\frac{r}{K}N^2+\left(r-\frac{\beta b}{m+\mu}\right)M+r N=0
\end{equation}
and $\psi_2 (M,N)=0$, given by
\begin{equation}\label{psi2}
(-p+\delta)MN-pN^2+\frac{\beta b}{m+\mu}M-nN=0.
\end{equation}
Thus, the equilibrium follows from determining the intersection of these two conic
sections in the fist quadrant of the $M-N$ phase plane.

We begin by classifying them: let us focus on the first one.

Now \eqref{psi1} represents a non degenerate conic if and only if
\begin{equation}
\frac{\beta^2 b^2}{(m+\mu)^2}+\frac{\beta \delta b K}{m+\mu} \neq \delta Kr
\end{equation}
and under this assumption, computing the second invariant, we find that it has the value
\begin{equation*}
-\frac{4r \delta+\delta^2 K}{4K}<0,
\end{equation*}
for which $\psi_1$ is a hyperbola. Its center $C(x_c,y_c)$ is
\begin{align*}
x_c=\frac{r}{\delta (4r+\delta K )}\left( \delta K+\frac{2 \beta b}{m+\mu}\right)>0, \\
y_c=-\frac{1}{\delta (4r+\delta K)} \left(\frac{b \beta(2 r+\delta K )}{m+\mu}
-\delta Kr \right).
\end{align*}
The slopes of the asymptotes satisfy the quadratic equation
\begin{equation*}
-\frac{r}{K}m^2+\left(-\frac{2r}{K}-\delta \right)m-\frac{r}{K}=0,
\end{equation*}
which has the roots
\begin{align*}
m_{1,2}=
&\frac{-2r-K\delta \pm \sqrt{(2r+K \delta)^2-4 r^2}}{2r}<0.
\end{align*}

We also look for the points at which $\psi_1$ intersects respectively the $M$ and $N$ axes.
The hyperbola $\psi_1$ passes through the origin and crosses the vertical axis in its
at the positive height $N_1=K$ and the horizontal one at the abscissa
$$
M_1=\frac{K}{r}\left(r-\frac{\beta b}{m+ \mu} \right),
$$
which is positive if the inequality
\begin{equation}\label{psi1_r}
r>\frac{\beta b}{m+ \mu}
\end{equation}
holds.

To sum up, we plot the graph of $\psi_1$ in both cases.
The hyperbola $\psi_1$ has a branch in the positive quadrant joining
$(0,K)$ and $(M_1,0)$ if (\ref{psi1_r}) holds, Case 1,
and instead joining the origin with $(0,K)$ otherwise, Case 2.
\begin{figure}[h!]
\centering \textbf{Hyperbola $\psi_1$} \par \medskip
\subfigure
{\includegraphics[scale=.4]{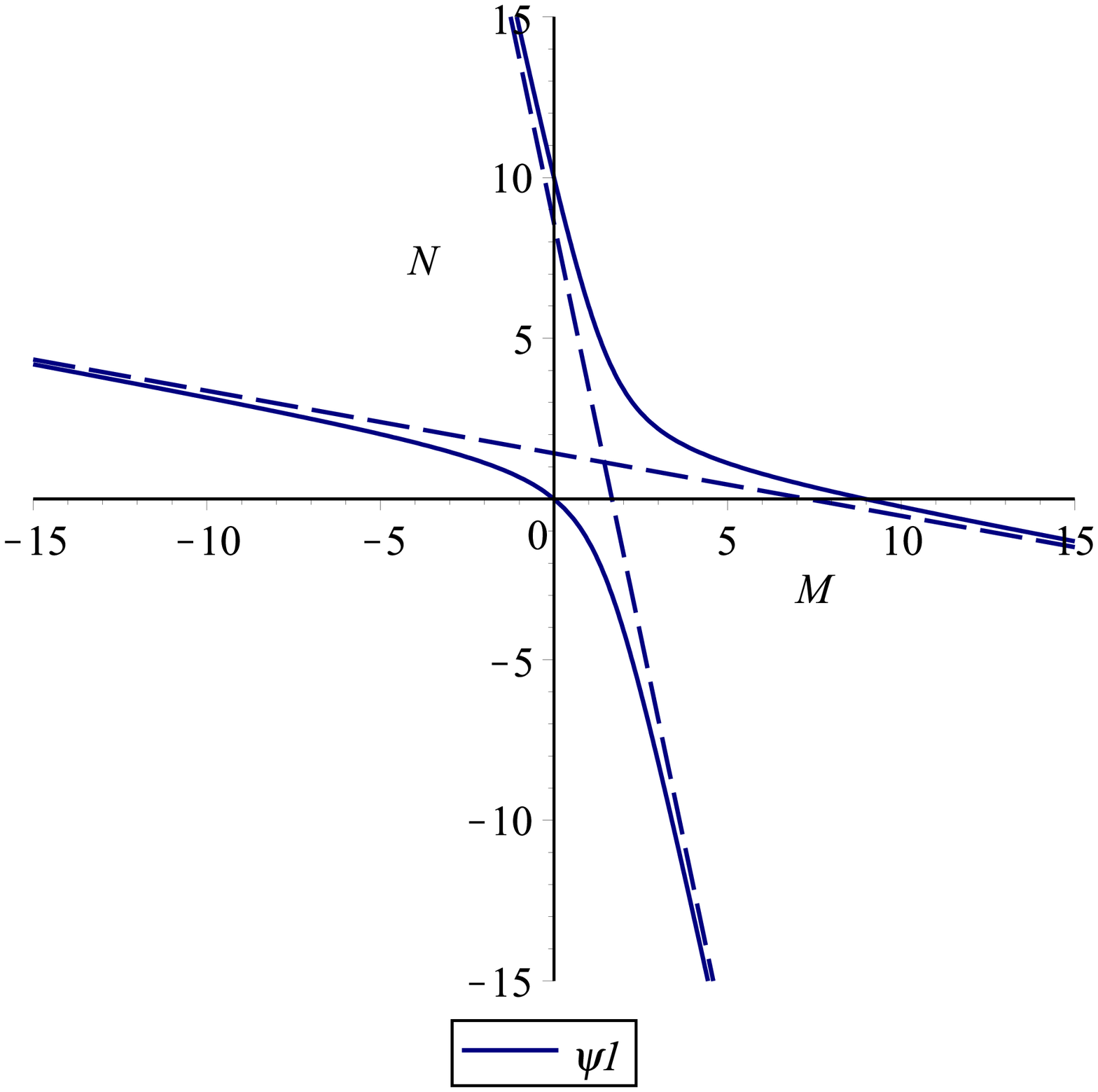}}
\caption{Case 1. $r>\frac{\beta b}{m+ \mu}$, for the parameters value $r=0.06, e=0.001, \lambda=0.03, b=250, K=10, \beta=0.0002, m=0.04, \mu=8, \delta=0.02, n=0.007, p=0.08$.}
\label{psi1a}
\hspace{3mm}
\subfigure 
{\includegraphics[scale=.4]{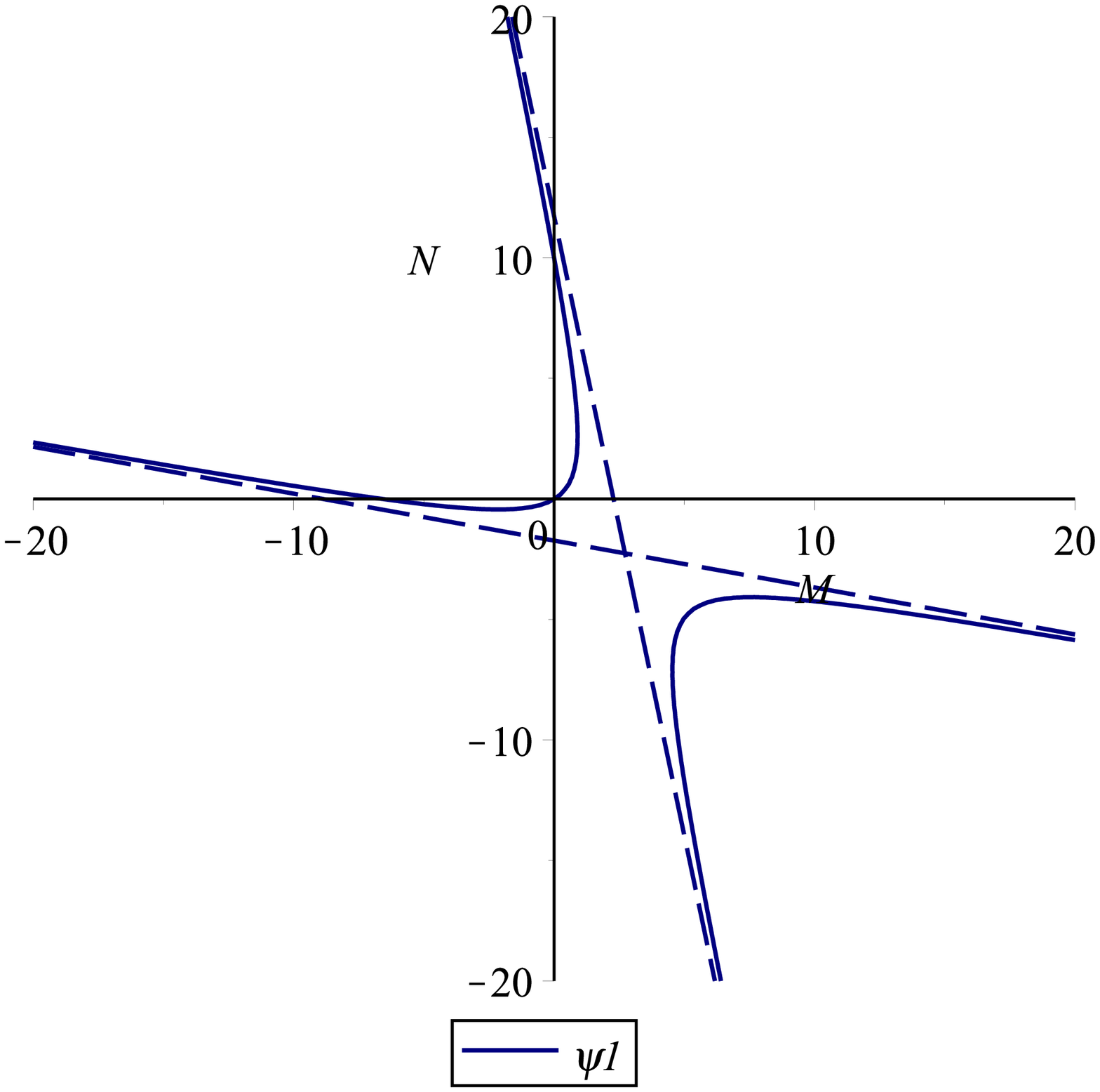}}
\caption{Case 2. $r<\frac{\beta b}{m+ \mu}$, for the parameters value $r=0.06, e=0.001, \lambda=0.03, b=1000, K=10, \beta=0.001, m=0.04, \mu=10, \delta=0.02, n=0.007, p=0.08$.}
\label{psi1b}
\end{figure}

\clearpage
 
\begin{oss}
Observe that the change in the two different shapes of the hyperbola $\psi_1$
occurs when the parameter value $r$ crosses the threshold value
\begin{equation}
r_{deg}=\frac{\beta^2 b^2}{\delta K(m+\mu)^2}+\frac{\beta b}{m+\mu},
\end{equation}
for which the hyperbola degenerates into its asymptotes.
However, this has no consequence on the analysis results.
\end{oss}

For the second conic section, the first invariant, assumed to
be nonvanishing, is
\begin{equation*}
\frac{(p-\delta)n\beta b}{4(m+\mu)}+\frac{p \beta^2 b^2}{4(m+ \mu)^2} \neq 0,
\end{equation*}
while the second one is always negative,
\begin{equation*}
-\frac{(-p+\delta)^2}{4}<0,
\end{equation*}
showing that also $\psi_2$ is a hyperbola.
Its center $C'$ is
\begin{align*}
&x_{c'}=\frac{1}{(p-\delta)^2}\left[ n(\delta-p)-\frac{2p \beta b}{m+ \mu} \right], \quad
&y_{c'}=\frac{\beta b}{(p- \delta)(m+ \mu)}.
\end{align*}
Here, the slopes of the asymptotes solve the quadratic equation
\begin{equation*}
(\delta -p)l-pl^2=0,
\end{equation*}
which provide immediately
\begin{equation*}
l_1=0, \quad l_2=\frac{\delta - p}{p},
\end{equation*}
so that $\psi_2$ has the horizontal asymptote
\begin{equation*}
N=\frac{\beta b}{(p- \delta)(m+ \mu)}.
\end{equation*}
Further, the hyperbola $\psi_2$ passes through the origin and crosses the
negative  vertical axis.
There are two possible cases, $p>\delta$ and $p<\delta$.
In the former, Case A, the feasible branch raises up monotonically
from the origin to the horizontal asymptote, in the latter, Case B, instead
it raises up again from the origin to the oblique asymptote.

\begin{figure}[h!]
\centering \textbf{Hyperbola $\psi_2$} \par \medskip
\subfigure
{\includegraphics[scale=.4, width = .5 \textwidth]{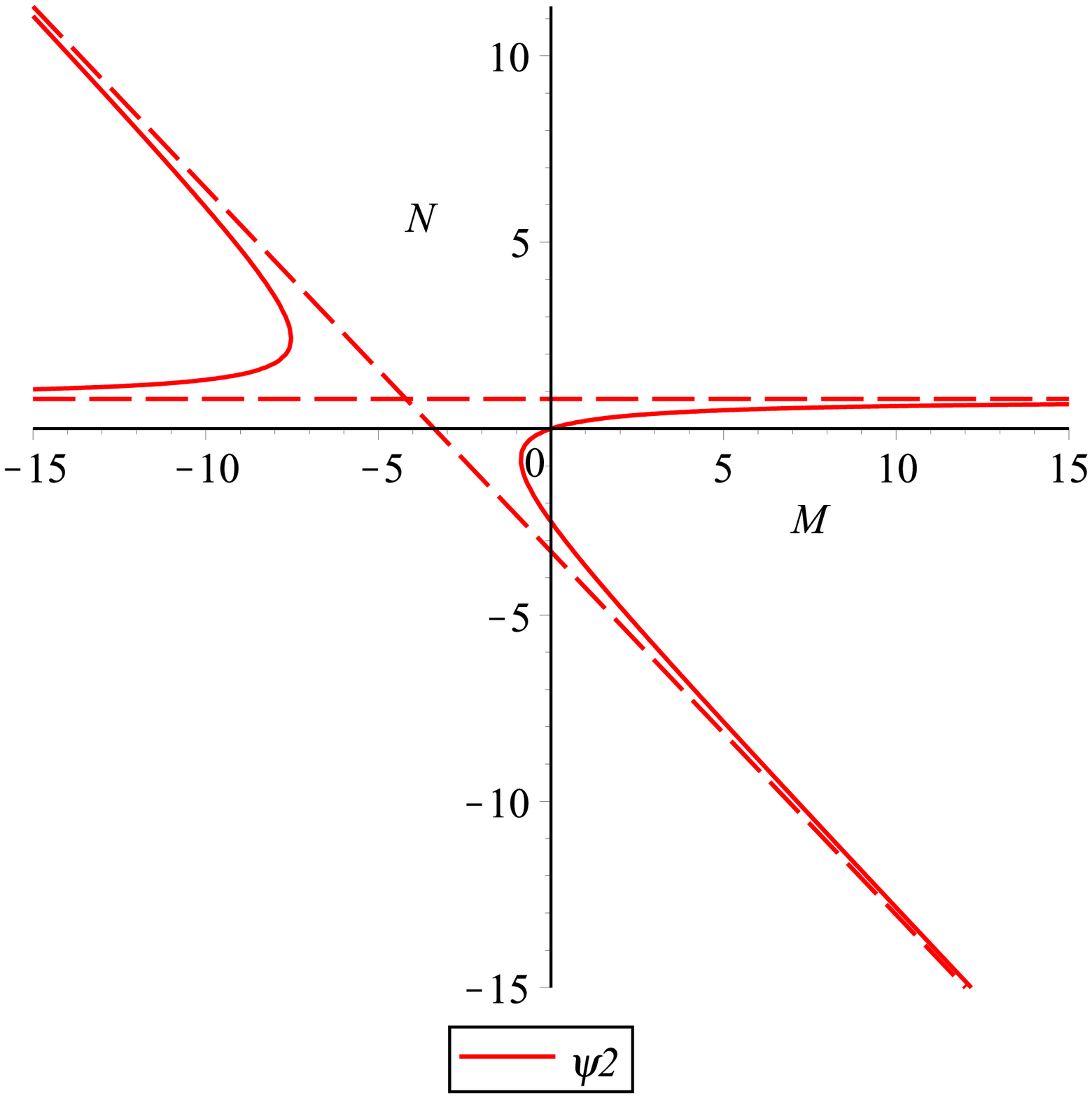}}
\caption{Case A. $p>\delta$, for the parameters value $r=0.06, e=0.001, \lambda=0.03, b=250, K=10, \beta=0.002, m=0.04, \mu=8, \delta=0.002, n=0.2, p=0.08.$.}
\label{psi2A}
\hspace{3mm}
\subfigure
{\includegraphics[scale=.4, width = .5 \textwidth]{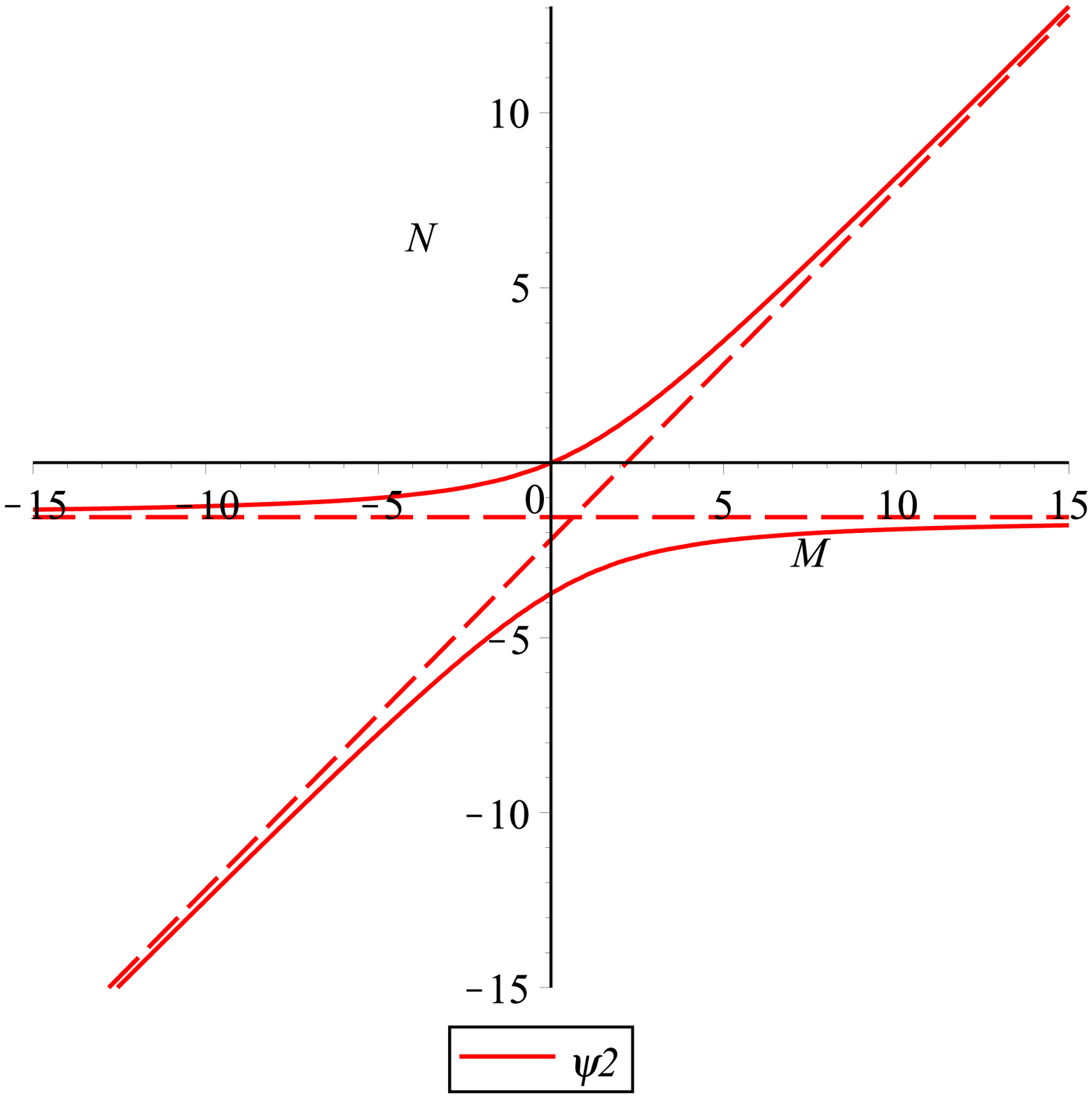}}
\caption{Case B. $p<\delta$, for the parameters value $r=0.06, e=0.001, \lambda=0.03, b=250, K=10, \beta=0.002, m=0.04, \mu=8, \delta=0.08, n=0.15, p=0.04$.}
\label{psi2B}
\end{figure}

\begin{oss}
The change occurs when the hyperbola degenerates into its asymptotes,
as the parameter value $\delta$ crosses the critical threshold value
\begin{equation}
\delta_{deg}=p+\frac{\beta b p}{n(m+\mu)}.
\end{equation}
But, just as for $\psi_1$, this change in shape does not affect the analysis results.
\end{oss}

Evidently, the hyperbolae $\psi_1$ and $\psi_2$ meet at the origin and they
further intersect at another point at least, which we require it to be located in
the first quadrant for feasibility.
The following cases may arise.

\begin{itemize}
\item[Case $1$A.] 

If $r>\frac{\beta b}{m+\mu}$ and $p>\delta$, there is a feasible intersection point.

\item[Case $1$B.]

If
$r>\frac{\beta b}{m+\mu}$ and $p<\delta$,
there is a feasible intersection point.

\item[Case $2$A.] 

If 
$r<\frac{\beta b}{m+\mu}$ and $p>\delta$,
two possible situations may occur,
the hyperbolae intersect at a feasible point if and only if the slope at the origin
of $\psi_2$ is larger than that of $\psi_1$.

In this case, the hyperbolae intersect at a feasible point if the slope at the origin
of $\psi_2$ is larger than that of $\psi_1$.
The slopes of the two hyperbolae at the origin to be
\begin{equation}
N'_{\psi_1}(0)=-1+\frac{\beta b}{r(m+\mu)}, \quad
N'_{\psi_2}(0)=\frac{\beta b}{n(m+\mu)}.
\end{equation}
Hence, the condition that they must satisfy is $N'_{\psi_2}(0)>N'_{\psi_1}(0)$, which becomes
\begin{equation}\label{slopes}
\beta b(r-n)+rn(m+\mu)>0.
\end{equation}

\item[Case $2$B.] 

If 
$r<\frac{\beta b}{m+\mu}$ and $p<\delta$,
two possible situations may occur, again depending on the slopes at the origin of
the two conic sections, which give once again condition (\ref{slopes}).

\end{itemize}

In summary, we can conclude that the equilibrium point $E_5$ is feasible if 
(\ref{psi1_r}) holds, or otherwise, but then (\ref{slopes})
is satisfied.

The coexistence equilibrium will finally be analysed numerically.

\section{Stability}\label{sec:stab}

The Jacobian matrix for the system \eqref{model} at a generic point is given by
\begin{equation}\label{Jac}
J=
\left( \begin{array}{cccc}
J_{11} & -\frac{bB}{(B+I)^2}- \gamma B & 0 & -\lambda B  \\
-\frac{bI}{(B+I)^2}+ \lambda N + \gamma I &J_{22} & 0 & \lambda B   \\
-eM & - \beta M  & J_{33} &  J_{34}\\
-eN & \beta M  & J_{43}  & J_{44}   \\
\end{array}\right),
\end{equation}
where
\begin{align*}
&J_{11}= \frac{b}{B+I}-\frac{bB}{(B+I)^2}-\lambda N-m - \gamma I \\
&J_{22}=\frac{b}{B+I}-\frac{bI}{(B+I)^2}-m - \mu +\gamma B  \\
&J_{33}=-\frac{r(M+N)}{K}+r(1-\frac{M+N}{K}) - \beta I -eB-\delta N \\
&J_{34}=-\frac{r(M+N)}{K}+r(1-\frac{M+N}{K})-\delta M  \\
&J_{43}=\delta N -pN + \beta I  \\
&J_{44}=-n-eB+\delta M-p(M+N)-pN.
\end{align*}

We now turn to the stability analysis of each equilibrium. However,
in view of the difficulty of analytically assessing even its existence, as already stated above,
the stability coexistence equilibrium $E^*$ is investigated by means of simulations.

\subsection{Stability of $E_1$}
For the equilibrium point $E_1$
the eigenvalues are 
$$\Lambda_1 = -m, \quad
\Lambda_2 = \frac{b \gamma - m \mu}{m}, \quad
\Lambda_3 = \frac{mr-be}{m}, \quad
\Lambda_4 = -\frac{be+mn}{m}.
$$
Therefore, the equilibrium $E_1$ is stable if and only if 
\begin{equation}\label{E1stab}
\gamma <\frac m {b \mu },
\quad
r<\frac{be}{m}.
\end{equation}

\subsection{Stability of $E_2$}

Two of the eigenvalues of $J$ evaluated at $E_2$ are immediately found,
$$
\Lambda_1=\mu-\frac{\gamma b}{m+\mu}=\frac{m \mu +\mu^2 - \gamma b}{m+\mu}, \quad
\Lambda_2=-m-\mu.
$$
The other two eigenvalues are the roots of the equation 
\begin{equation*}
\Lambda^2 - \textrm{tr}\widehat J^*(E_2) \Lambda + \det \widehat J^*(E_2)=0,
\end{equation*}
where $\textrm{tr}\widehat J^*(E_2)$ and $\det \widehat J^*(E_2)$
are respectively the trace and the determinant of the submatrix 
\begin{align*}
\widehat J^*(E_2)=
\left( \begin{array}{cc}
r-\dfrac{b \beta}{m+ \mu} & r   \\
\dfrac{\beta b}{m+\mu}  & -n   \\
\end{array}\right),
\end{align*}
obtained from \eqref{Jac} deleting the first two rows and the first two columns.
The Routh-Hurwitz criterion ensures that the eigenvalues
have negative real part if and only if
\begin{equation}\label{trdetE1}
\textrm{tr}\widehat J^*(E_2)=r-\frac{b \beta}{m+ \mu}-n<0, \quad
\det \widehat J^*(E_2)=\frac{\beta b n}{m+\mu}-rn-\frac{r \beta b}{m+\mu}>0,
\end{equation}
which can be rewritten as
\begin{equation}\label{trdetE1_short}
(r-n)(m+\mu) < b \beta, \quad
b \beta(n-r) > rn(m+\mu).
\end{equation}
Clearly, we need $n>r$ for the second condition and this implies that the first one holds.
Furthermore, for stability also the first eigenvalue $\Lambda_1$ must be negative, so the following condition is also necessary, in addition to \eqref{trdetE1_short}.
Therefore, at $E_2$ stability occurs if and only if
\begin{equation}\label{E2stab}
b \gamma - m \mu -\mu^2>0,
\quad
b\beta> (m+\mu) \frac{rn}{n-r}, \quad n>r.
\end{equation} 

\subsection{Stability analysis for $E_3$}

Also for the equilibrium point $E_3$ two eigenvalues are explicit,
$\Lambda_1=-m$, $\Lambda_2=-r+\frac{be}{m}<0$ by the feasibility condition (\ref{E3feas}).
The other two eigenvalues are those of the remaining minor of the Jacobian $\widetilde J^*(E_3)$.
They are the roots of the quadratic 
\begin{equation*}
\Lambda^2 - {\textrm{tr}} \widetilde J^*(E_3) \Lambda + \det \widetilde J^*(E_3)=0.
\end{equation*}
Again, by the Routh-Hurwitz criterion, at $E_3$ stability is achieved for
\begin{eqnarray}\label{E3stab}
\frac{\gamma b}{m}+ \left( \frac{\delta K}{r}
-\frac{pK}{r} \right) \left( r- \frac{be}{m} \right) < \mu +n +\frac{be}{m}
\qquad \qquad \\ \nonumber
\left(\frac{\gamma b}{m} -\mu \right) \left[ \left( \frac{\delta K}{r}-\frac{pK}{r}\right)
\left(r-\frac{be}{m}\right) -n-\frac{be}{m}\right] 
> \frac{\lambda b}{m} \frac{bK}{r} \left(r- \frac{be}{m} \right) .
\end{eqnarray}  

\subsection{Stability analysis for $E_4$}

The characteristic equation at $E_4$ factorizes into the product of two quadratic equations.
The first is
\begin{equation}
\label{first}
\Lambda^2 - {\textrm{tr}} J^*(E_4) \Lambda + \det J^*(E_4)=0,
\end{equation}
where $J^*(E_4)$ is the submatrix 
\begin{align*}
J^*(E_4)=
\left( \begin{array}{cc}
\dfrac{b \gamma (b \gamma -m \mu -\mu^2)}{\mu^3} &
\dfrac{(b \gamma -m\mu-\mu^2)(b\gamma+\mu^2)}{\mu^3}   \\
-\dfrac{(b\gamma-m\mu)(b\gamma-\mu^2)}{\mu^3} &
- \dfrac{b \gamma (b \gamma -m \mu)}{\mu^3}  \\
\end{array}\right).
\end{align*}

Thus, the Routh Hurwitz conditions give the first set of inequalities needed for stability,
the first one of which holds unconditionally:
\begin{align}\label{tracciaE3}
{\textrm{tr}} J^*(E_4)=&\frac{b \gamma (b \gamma -m \mu -\mu^2)}{\mu^3}- \frac{b \gamma (b \gamma -m \mu)}{\mu^3}=-\frac{b\gamma}{\mu} <0, \\
\label{detE3}
\begin{split}
\det J^*(E_4)=&\left( \frac{b \gamma (b \gamma -m \mu -\mu^2)}{\mu^3}\right)
\left( -\frac{b \gamma (b \gamma -m \mu)}{\mu^3}\right) \\
&+\frac{(b \gamma -m\mu-\mu^2)(b\gamma+\mu^2)(b\gamma-m\mu)(b\gamma-\mu^2)}{\mu^6}>0.
\end{split}
\end{align}
Dividing \eqref{detE3} by $(b \gamma -m\mu-\mu^2)(b\gamma-m\mu)$, which is negative from
the feasibility condition (\ref{E3feas}), the inequality is reduced to
\begin{equation*}
-\frac{1}{\mu^2}<0,
\end{equation*}
so that it always holds.
The second quadratic is
\begin{equation}\label{second}
\Lambda^2 - {\textrm{tr}}J_*(E_4) \Lambda + \det J_*(E_4)=0,
\end{equation}
where
\begin{align}
J_*(E_4)=
\left( \begin{array}{cc}
r-\dfrac{\beta(b\gamma-m\mu)+e(-b\gamma+m\mu+\mu^2)}{\mu\gamma} & r   \\
\dfrac{\beta(b\gamma-m\mu)}{\mu\gamma}  & \dfrac{e(b\gamma-m\mu-\mu^2)}{\mu\gamma}-n    \\
\end{array}\right).
\label{J_*E3}
\end{align}
Once again, the associated Routh-Hurwitz conditions
\begin{align*}
\begin{split}
{\textrm{tr}}J_*(E_4)
&=r-\frac{\beta(b\gamma-m\mu)}{\mu\gamma}+\frac{2e(b\gamma-m\mu-\mu^2)}{\mu \gamma}-n <0, \end{split}\\
\begin{split}
\det J_*(E_4)=&\left(r-\frac{\beta(b\gamma-m\mu)+e(-b\gamma+m\mu+\mu^2)}{\mu\gamma}\right)\left(\frac{e(b\gamma-m\mu-\mu^2)}{\mu\gamma}-n\right) \\
&-\frac{r\beta(b\gamma-m\mu)}{\mu\gamma}>0,
\end{split}
\end{align*}
upon rearranging, provide the following stability conditions:
\begin{eqnarray}\label{E4stab}
r+\frac{2e(b\gamma-m\mu-\mu^2)}{\mu \gamma} < n+\frac{\beta(b\gamma-m\mu)}{\mu\gamma},\\
\nonumber
\left(r-\frac{\beta(b\gamma-m\mu)+e(-b\gamma+m\mu+\mu^2)}{\mu\gamma}\right)
\left(\frac{e(b\gamma-m\mu-\mu^2)}{\mu\gamma}-n\right) \\ \nonumber
> \frac{r\beta(b\gamma-m\mu)}{\mu\gamma}.
\end{eqnarray}

\subsection{Stability analysis for $E_5$}

For the equilibrium point $E_5$ the analytic expression of the mite populations
$M_5$ and $N_5$ is not know.
Anyway, we can evaluate the Jacobian matrix \eqref{Jac} for $B=0$ and $I=\frac{b}{m+\mu}$.
The characteristic equations factorizes to provide immediately two eigenvalues, namely
$$
\Lambda_1=\mu-\lambda N-\frac{\gamma b}{m+\mu}, \quad
\Lambda_2=-m-\mu,
$$
yielding the first stability condition 
\begin{equation}\label{E5stab1}
N>\frac{1}{\lambda}\left(\mu-\frac{\gamma b}{m+\mu} \right).
\end{equation}

The other two eigenvalues are the roots of 
\begin{equation*}
\Lambda^2 - {\textrm{tr}}\overline J^*(E_5) \Lambda + \det \overline J^*(E_5)=0
\end{equation*}
where $\overline J^*(E_5)$ is the submatrix 
\begin{align*}
\overline J^*(E_5)=
\left( \begin{array}{cc}
\overline J_{33} & r-\frac{2r}{K}(M+N)-\delta M    \\
\delta N-pN+\frac{\beta b}{m+\mu}  & \overline J_{44} \\
\end{array}\right),
\end{align*}
obtained from \eqref{Jac} deleting the first two rows and columns and
$$
\overline J_{33}=r-\frac{2r}{K}(M+N)-\frac{\beta b}{m+\mu}-\delta N, \quad
\overline J_{44}=-n+(\delta -p)M-2pN.
$$
The Routh-Hurwitz criterion provide the remaining stability conditions
\begin{eqnarray}\label{E5stab2}
M\left( -\frac{2r}{K}+\delta -p \right)+N \left(-\frac{2r}{K}-\delta -2p \right)+r
< n + \frac{\beta b}{m+\mu}, \\ \nonumber
\left[ r-\frac{2r}{K}(M+N)-\frac{\beta b}{m+\mu}-\delta N\right] \left[-n+(\delta -p)M-2pN \right] \\ \nonumber
> \left[r-\frac{2r}{K}(M+N)-\delta N \right] \left[(\delta -p)N+\frac{\beta b}{m+\mu} \right].
\end{eqnarray}

\section{Numerical simulations and interpretation}\label{sec:sim}
In this section, after describing the set of parameter values used,
we carry out numerical simulations using
the built-in ordinary differential equations solver ode45.

\subsection{The model parameters}

We take some model parameters from the literature, \cite{joyce},
\cite{parametri}, \cite{maci} and \cite{api2}. 

Taking the day as the time unit,
the birth rate of worker honey bees is taken as $b=2500$,
while their natural mortality rate is $m=0.04$, equivalent to a life expectancy of
$25$ days, \cite{joyce}.

The literature does not provide a precise value for the grooming behavior,
but from  \cite{joyce} a range of reasonable values of $e$ lies between $10^{-6}$ and $10^{-5}$.

During the bee season, i.e. spring and summer, the mite population doubles every month,
so we fix $r \approx \ln 2 \div 30$, that is $r=0.02$, \cite{api2}.

Their carrying capacity is taken as $K=15000$ mites, \cite{parametri},
while their natural mortality rate in the phoretic phase is estimated to be
$n=0.007$ during the bee season, \cite{maci}.   

Table \eqref{table3} summarizes these values, that are used in the numerical simulations.

\begin{table}[h!]\footnotesize
\begin{center}
\begin{tabular}{|c|c|c|c|c|}
\hline
\textbf{Parameter}     &\textbf{Interpretation}        &\textbf{Value}  &\textbf{Unit}  &\textbf{Source}  \\ \hline
\textbf{$b$}            &Bee daily birth rate          &$2500$       &$bees \times day^{-1}$  &\cite{joyce} \\   \hline 
\textbf{$e$}            &Grooming rate of healthy bee     &$10^{-6} - 10^{-5}$ &$day^{-1}$   &\cite{joyce}      \\   \hline 
\textbf{$m$}        &Bee natural mortality rate    &$0.04$ &$day^{-1}$ &\cite{joyce}    \\   \hline
\textbf{$K$}          &\textit{Varroa} carrying capacity     &$15000$ &$mites$ &\cite{parametri}    \\   \hline
\textbf{$r$}         &\textit{Varroa} growth rate    &$0.02$ &$day^{-1}$ &\cite{api2}    \\   \hline
\textbf{$n$}       &\textit{Varroa} natural mortality &$0.007$ &$day^{-1}$ &\cite{maci} \\ 
&rate in the phoretic phase & & & \\ \hline
\end{tabular}
\caption{Model parameters}
\label{table3}
\end{center}
\end{table}

\subsection{Discussion}

The initial conditions for the colony are given from field data.
All the bees are healthy since the infected,
having a lower longevity, cannot survive the winter.
Further, the acaricide treatments against \textit{Varroa} performed
in the autumn and winter theoretically allow to eradicate the infestation,
so that the mite population usually does not exceed $10$ units at the beginning of the bee season.

\begin{table}[h!]
\begin{center}
\begin{tabular}{|c|c|}
\hline
\textbf{Population}     &\textbf{Value}   \\ \hline
Sound bees            &$15000$     \\   \hline 
Infected bees           &$0$    \\   \hline 
Mites        &$\leq 10$      \\   \hline
\end{tabular}
\caption{Initial condition for simulations}
\label{table4}
\end{center}
\end{table}

For the parameters $b, e, m, K, r, n$ we use the reference values
from the real situation described in the previous section, see Table \eqref{table3}.
The remaining parameter values are arbitrarily chosen to simulate a hypothetical hive.

\subsection{Equilibrium $E_1$}

At the equilibrium $E_1$ only the healthy bees survive.
It represents the best situation for the hive. 
Healthy and infected mites are wiped out while the healthy bee population
increases to a level of $62500$ units.

The equilibrium $E_1$ is stably attained for the following choice of parameters:
$r=0.02$, $K=15000$, $b=2500$, $m=0.04$, $n=0.007$, $e=0.00001$, $\mu=8$,
$\lambda=0.3$; $\delta=0.0001$, $\beta=0.1$, $\gamma=0.0001$, $p=0.04$.
Initial conditions $B=15000$, $I=0$, $M=3$, $N=2$.

The second stability condition \eqref{E1stab} is always satisfied by the field data.
Infected bees become extinct in view of the remaining stability condition \eqref{E1stab},
that, rewritten in the form $\gamma< m \mu b^{-1}$,
establishes an upper bound on the horizontal transmission rate of
the virus among bees. 

\subsection{Equilibrium $E_2$}
It is interesting to note that
the equilibrium point $E_2$, in which the disease affects the whole bee colony
but the mites disappear, is unstable when the field data are used,
since the third stability condition \eqref{E2stab}
is not satisfied. Indeed note that $n-r=0.007-0.02<0$.

\subsection{Equilibrium $E_3$}
Also 
the disease-free equilibrium point $E_3$, in which healthy mites are present,
is infeasible in the field conditions, since the feasibility condition \eqref{E3feas}
is not satisfied.

This result is well substantiated by the empirical beekeepers observations
in natural honey bee colonies.
Namely, the viral infection occurs whenever the mite population is present.

\subsection{Equilibrium $E_4$}
At the mite-free equilibrium point $E_4$
the disease remains endemic among the bees.

Rearranging \eqref{E4stab}, we can see that in this situation the opposite condition of equilibrium $E_1$ must be verified, namely
\begin{equation}\label{E4stab_new}
\frac{m \mu}{b}< \gamma < \frac{\mu^2 + m \mu}{b}
\end{equation}
imposing a lower and an upper bound on
the disease horizontal transmission rate among bees.  
It is indeed reasonable to expect that also the population of infected bees survives,
as opposed to the equilibrium $E_1$, with a high enough contact rate.

If the horizontal transmission rate becomes too large though,
exceeding its upper bound, the point $E_4$ becomes unstable,
but the first stability condition for $E_2$, \eqref{E1stab}, holds.
Apparently we could obtain the equilibrium $E_2$ in which the infection affects
all the honey bees, but we already know that $E_2$ is always unstable from the field data
and therefore the infected bees in such situation do not thrive alone in the hive.

The equilibrium $E_4$ is obtained for the parameters values: $r=0.02$, $K=15000$, $b=2500$,
$m=0.04$, $n=0.007$, $e=0.00001$,$\mu=6$, $\lambda=0.01$; $\delta=0.0005$, $\beta=0.00006$,
$\gamma=0.0002$, $p=0.001$ and the initial conditions $B=15000$, $I=0$, $M=3$, $N=2$.
The eigenvalues of the Jacobian matrix for these parameter values are
$$
\Lambda_1 \approx -0.0416+0.5 i \quad
\Lambda_2 \approx -0.0416-0.5 i \quad
\Lambda_3 \approx -0.28 \quad
\Lambda_4 \approx -0.31.
$$
showing that there are damped oscillations.

\subsection{Equilibrium $E_5$}
The equilibrium $E_5$ represents the \textit{Varroa} invasion scenario.
It is always feasible by the field data, since \eqref{slopes}
is always satisfied, as its left hand side is $r-n=0.02-0.007>0$.
It can be obtained with the following choice for
the parameters: $r=0.02$, $K=15000$, $b=2500$, $m=0.04$, $n=0.007$, $e=0.00001$,$\mu=5$,
$\lambda=0.004$; $\delta=0.0005$, $\beta=0.0001$, $\gamma=0.3$, $p=0.0003$
and the initial conditions $B=15000$, $I=0$, $M=3$, $N=2$.

Also in this case, the viral infection remains endemic.
The effective virulence of the disease is increased by the presence of mites,
since the virus is directly vectored into the hemolymph of the honey bees.
As a consequence, the healthy bee population becomes extinct.

Since $E_2$ is unstable in view of the field data,
it seems much more likely that an epidemic could affect all the bees in a
\textit{Varroa}-infested colony than in a \textit{Varroa}-free colony.

In fact, the mite is a far more effective vector than nurse bees are.
Further, detailed studies show that most of honey bee viruses are present
at low levels but do not cause any apparent symptoms in \textit{Varroa}-free
colonies. Only important stressors, such as the introduction of
\emph{Varroa destructor}, can convert the silent infection into a symptomatic one,
\cite{GA}.

\subsection{The coexistence equilibrium $E^*$}

At coexistence the honey bees share their infected hive with mites.

Figures \eqref{sim_coes1sana} and \eqref{sim_coes2malata} show
two different levels of infection.
The spread of the virus throughout the colony depends on model parameters describing the
disease transmission, namely $\lambda$, $\gamma$, $\beta$, $\delta$. 
Reasonably, small values of these parameters lead to a rather higher prevalence of healthy
bees and mites with respect to the infected ones, Figure \eqref{sim_coes1sana}, while the
opposite situation is depicted in Figure \eqref{sim_coes2malata} in presence of higher
disease transmission rates. 
Furthermore, comparing the $\mu$-values in these simulations, we discover an interesting
and seemingly counter-intuitive phenomenon: a higher level of infection in the hive is
obtained for a lower value of the disease-related mortality and vice versa.
This can be interpreted by observing that a higher survival of infected bees makes
them longer available for the transmission of the virus and inevitably
leads to an increase of the viral titer in the colony.
Hence, when transmitted by \textit{Varroa} mites, the most virulent diseases
at the colony level become the least harmful for the single bees.
Thus the number of mites required for vectoring an epidemic spread
is higher for the most harmful viruses.
It has indeed been shown, for example, that the acute paralysis virus (APBV),
which is rapidly lethal to infected bees, requires a much larger mite population
for its outbreak of an epidemic than the mite population needed
for the deformed wing virus (DWV) outbreak,
which allows a greater survival of infected bees \cite{SM}. 

\begin{figure}[h!]
\centering 
{\includegraphics[scale=.7]{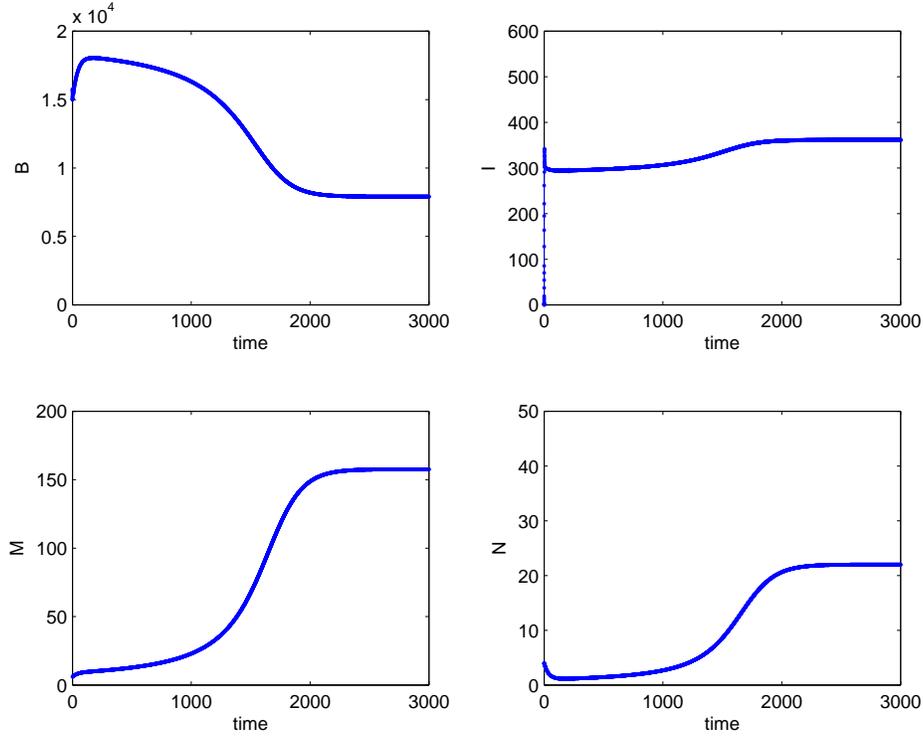}}
\caption{The coexistence equilibrium $E^*$, for the parameters values: $r=0.02$, $K=15000$,
$b=2500$, $m=0.04$, $n=0.007$, $e=0.000001$,$\mu=6$, $\lambda=0.007$; $\delta=0.0005$,
$\beta=0.00001$, $\gamma=0.0003$, $p=0.0005$. Initial conditions
$B=15000$, $I=0$, $M=6$, $N=4$.}
\label{sim_coes1sana}
\end{figure}

\begin{figure}[h!]
\centering 
{\includegraphics[scale=.7]{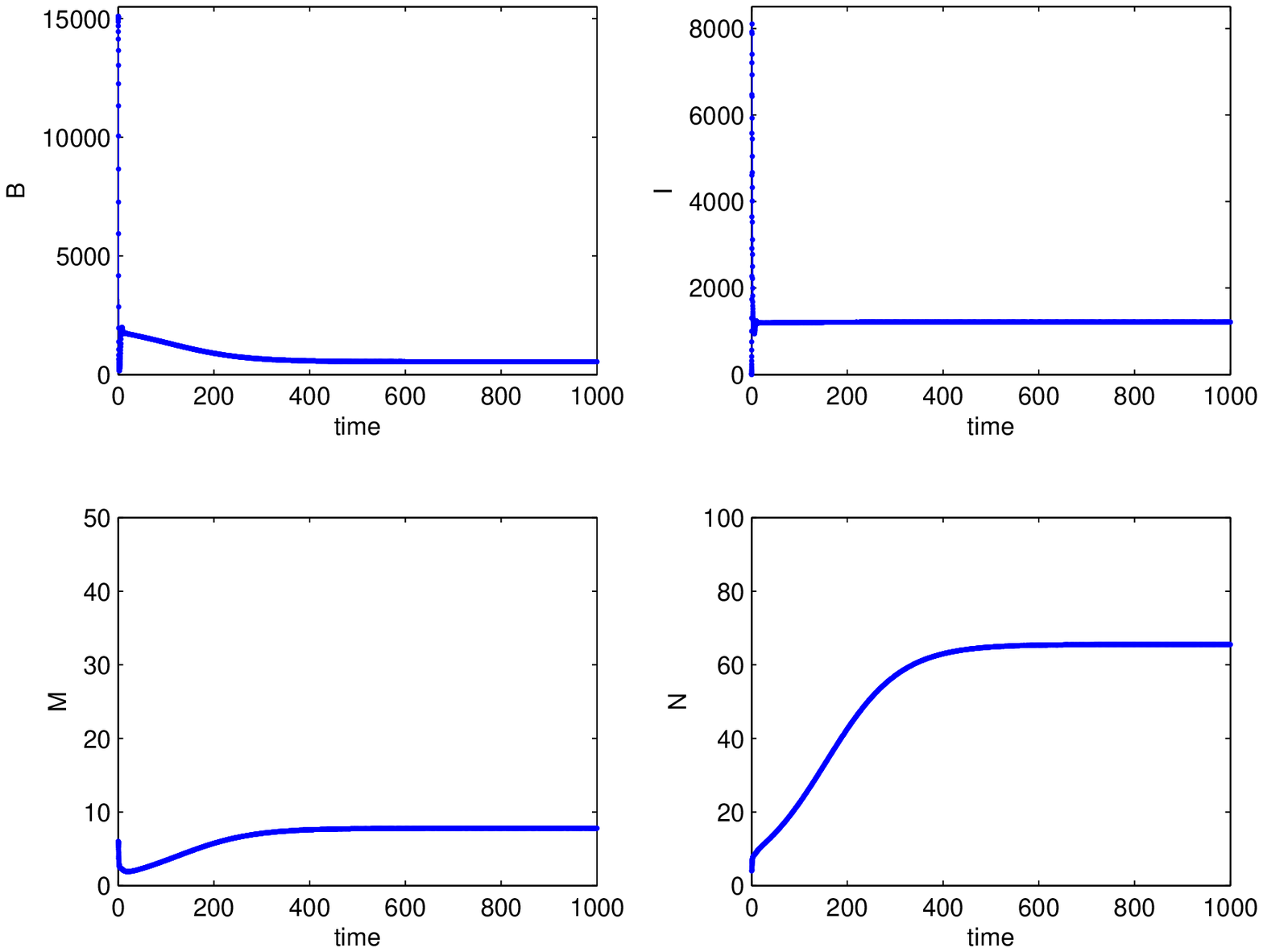}}
\caption{The coexistence equilibrium $E^*$, for the parameters values: $r=0.02$, $K=15000$,
$b=2500$, $m=0.04$, $n=0.007$, $e=0.000001$,$\mu=2$, $\lambda=0.01$; $\delta=0.001$,
$\beta=0.0001$, $\gamma=0.0006$, $p=0.0002$. Initial conditions $B=15000$, $I=0$, $M=6$,
$N=4$.}
\label{sim_coes2malata}
\end{figure}

Finally, in Figure \eqref{sim_coes_apidim_varrcr} we use a different initial condition from
the previous simulations to highlight the effect of \textit{Varroa} infestation on honey bee
viral infections: here the bee population drastically decreases as the mites increase. 
For a high value for the disease-related mortality of the bees ($\mu=7$), we observe a low
level of infection, with the infected bees not exceeding $350$,
although the decrease of healthy bees is quite dramatic.
    
\begin{figure}[h!]
\centering 
{\includegraphics[scale=.7]{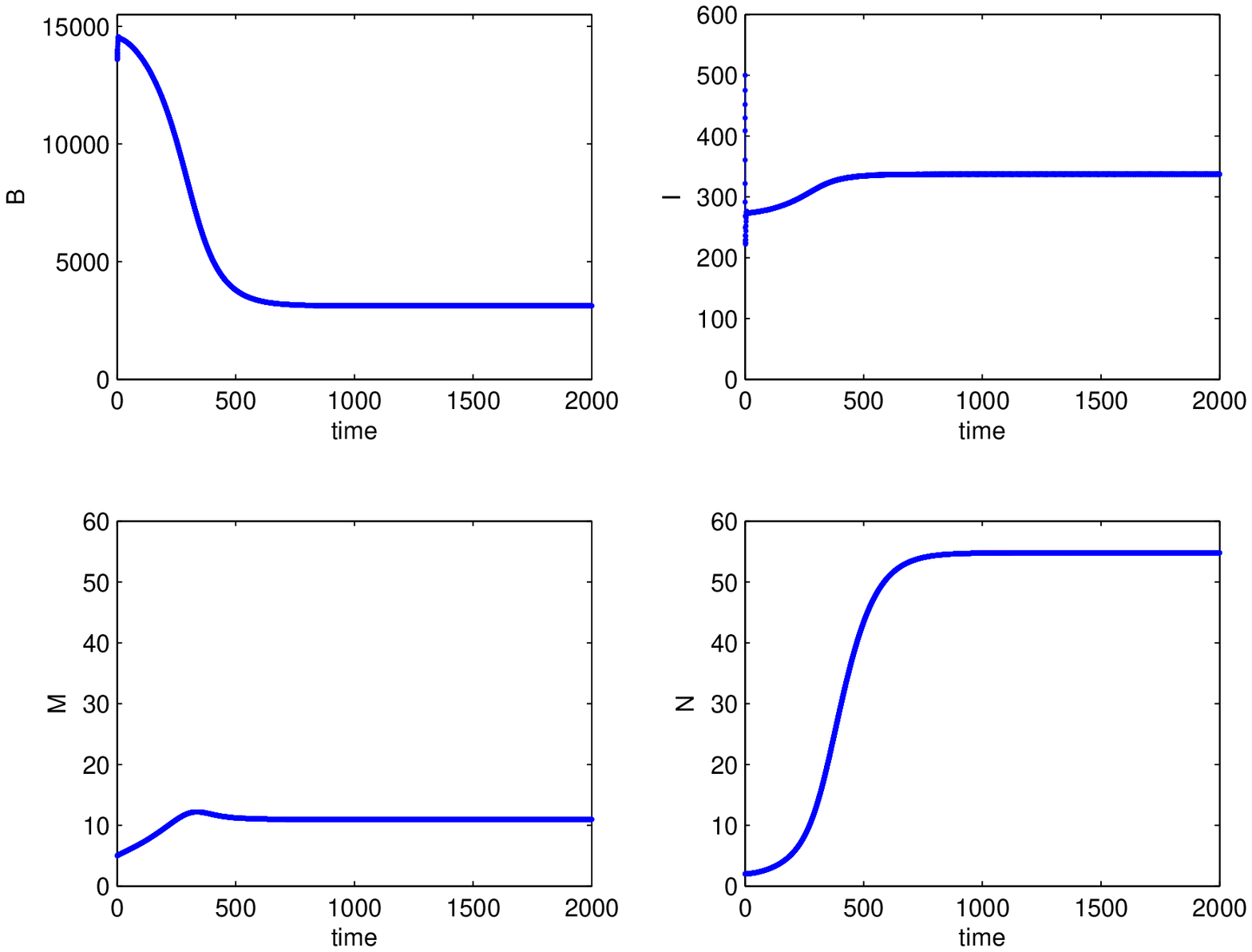}}
\caption{The coexistence equilibrium $E^*$, for the parameters values: $r=0.02$, $K=15000$,
$b=2500$, $m=0.04$, $n=0.007$, $e=0.000001$,$\mu=7$, $\lambda=0.01$; $\delta=0.002$,
$\beta=0.00002$, $\gamma=0.0004$, $p=0.0002$. Initial conditions $B=14000$, $I=500$, $M=5$,
$N=2$.}
\label{sim_coes_apidim_varrcr}
\end{figure}

\subsection{Summary}

We sum up the feasibility and stability conditions of the equilibrium points in Table
\ref{table2}.

\begin{table}[h!]
\begin{center}
\begin{tabular}{|c|c|c|c|}
\hline
\textbf{Equilibrium}     &\textbf{Feasibility}        &\textbf{Stability}  &
\textbf{In field} \\ &&& \textbf{conditions} \\ \hline
\textbf{$E_1$}            &always          &\eqref{E1stab} & allowed    \\   \hline 
\textbf{$E_2$}            &always          &\eqref{E2stab} & unstable     \\   \hline 
\textbf{$E_3$}         &\eqref{E3feas}      &\eqref{E3stab}& infeasible   \\   \hline
\textbf{$E_4$}         &\eqref{E4feas}    &\eqref{E4stab} &  allowed \\   \hline
\textbf{$E_5$}         &\eqref{psi1_r} or,
otherwise, \eqref{slopes}   &\eqref{E5stab1}, \eqref{E5stab2} &  allowed \\   \hline
\textbf{$E^*$}       & numerical & numerical & allowed \\ 
       & simulations & simulations & \\ \hline
\end{tabular}
\caption{Summarizing table of equilibria: feasibility and stability conditions}
\label{table2}
\end{center}
\end{table}

\section{Bifurcations}\label{sec:bif}

\subsection{Transcritical bifurcations}
Direct analytical considerations on the feasibility and stability conditions
are only possible between $E_1$ and $E_4$.
Some other transcritical bifurcations are found with the help of numerical simulations.

\begin{itemize}
\item $E_1$-$E_4$

Comparing \eqref{E1stab} and the left inequality in \eqref{E4feas},
we find that $E_4$ becomes feasible exactly when $E_1$ is unstable, and vice versa.
There is thus a transcritical bifurcation for which $E_4$ emanates from $E_1$ when
the bifurcation parameter $\gamma$ crosses from below the critical value
$\gamma^{\dagger}$ given by
\begin{equation}\label{gammatr}
\gamma^{\dagger} = \frac{m \mu}{b}.
\end{equation}
The simulation reported in Figure \eqref{transE0_E3} shows it explicitly.
For the chosen parameter values, $\gamma^{\dagger} = \frac{m \mu}{b}  \approx 0.000097$.

\begin{figure}[h!]
\centering 
{\includegraphics[scale=.7]{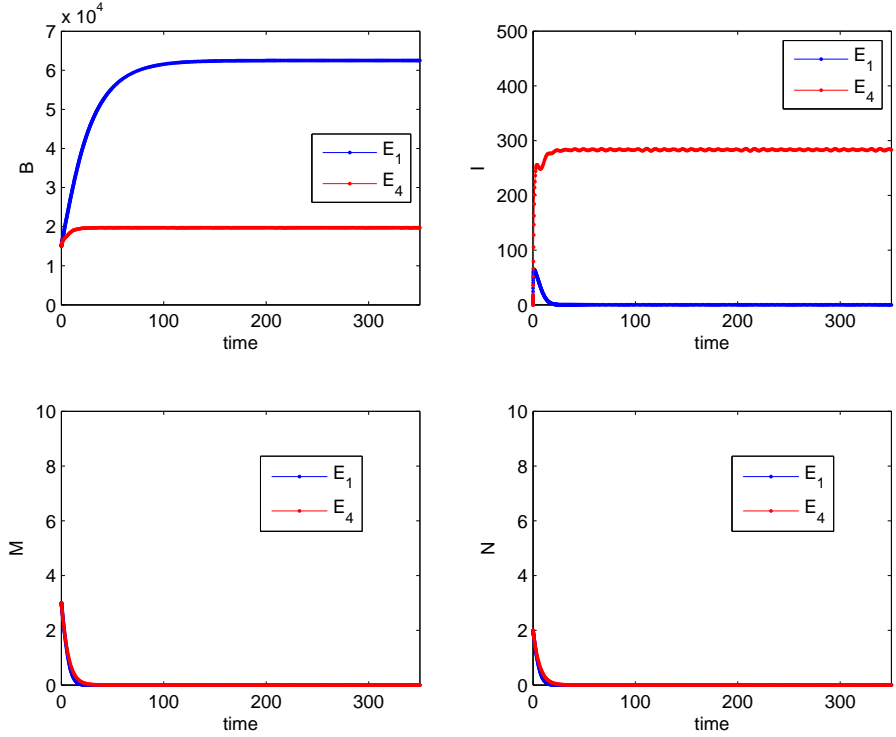}}
\caption{Transcritical bifurcation between $E_1$ and $E_4$, for the parameters values
$r=0.02$, $K=15000$, $b=2500$, $m=0.04$, $n=0.007$, $e=0.00001$,$\mu=6$, $\lambda=0.01$,
$\delta=0.0005$, $\beta=0.00006$, $\gamma=0.000094$ (blue), $\gamma=0.0003$ (red),
$p=0.001$. Initial conditions $B=15000$, $I=0$, $M=3$, $N=2$.}
\label{transE0_E3}
\end{figure}

\item $E_2$-$E_4$

Another transcritical bifurcation arises between $E_2$ and $E_4$, compare
the first \eqref{E2stab} and the right inequality in \eqref{E4feas}.
It occurs for the critical bifurcation parameter value
\begin{equation}\label{gammatr2}
\gamma^{\ddagger} = \mu \frac{m+\mu}{b}.
\end{equation}

\item $E_2$-$E_5$

Further, comparing the second \eqref{E2stab} and \eqref{slopes},
also between $E_2$ and $E_5$ we find anew a transcritical bifurcation.
In this case $E_5$ emanates from $E_2$ when the disease transmission rate between
infected bees and susceptible mites falls below the critical value
\begin{equation}\label{betatr}
\beta^{\dagger} = \frac{m+\mu}{b} \frac {rn}{n-r}.
\end{equation}

\item $E_4$-$E^*$

Numerically, a transition is seen to occur between the equilibrium $E_4$
and the coexistence point $E^*$. It is shown in Figure \eqref{transcE3-coes},
taking as bifurcation parameter $\mu$.

As we have observed in the previous section, a low value of the disease-related mortality
allows infected bees to survive longer. Instead, a higher value of $\mu$ causes a drastic
reduction of the infected bee population. As a consequence, although there is a harmful
epidemic, the bee population is not greatly affected at the colony level and mites are soon
groomed away by the healthy bees. 

\begin{figure}[h!]
\centering 
{\includegraphics[scale=.7]{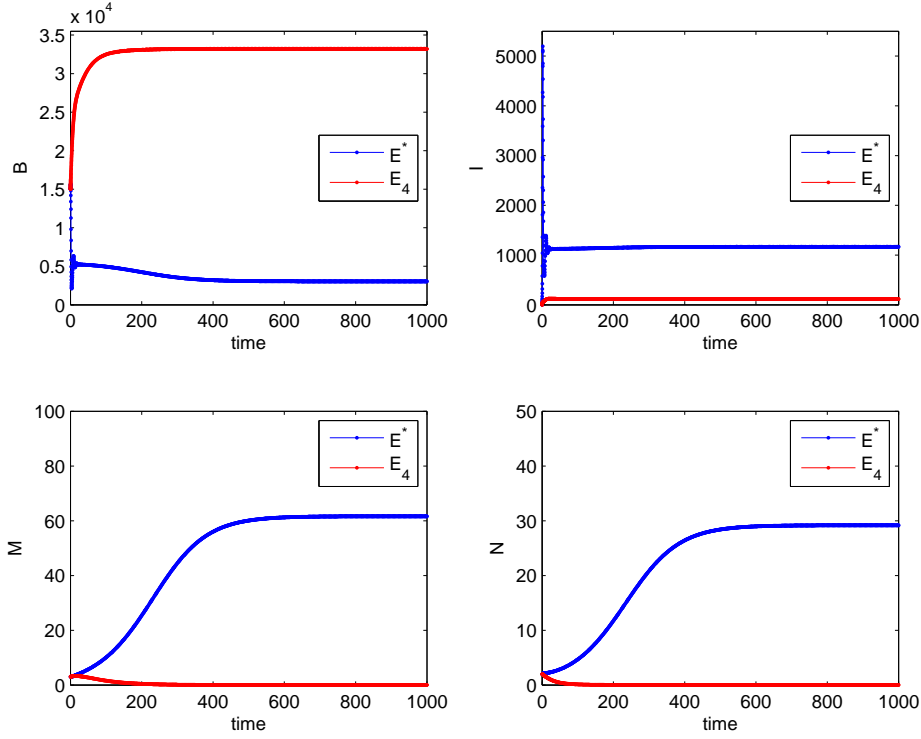}}
\caption{Transcritical bifurcation between $E_4$ and $E^*$, for the parameters values
$r=0.02$, $K=15000$, $b=2500$, $m=0.04$, $n=0.007$, $e=0.000001$,$\mu=2$ (blue), $\mu=10$
(red) $\lambda=0.007$, $\delta=0.0005$, $\beta=0.00001$, $\gamma=0.0003$, $p=0.0005$.
Initial conditions $B=15000$, $I=0$, $M=3$, $N=2$.} 
\label{transcE3-coes}
\end{figure}

\item $E_1$-$E_5$

In Figure \eqref{transcE0-E4}, another transcritical bifurcation is shown. For a small value
of $\gamma$ the system settles to the disease- and mite-free equilibrium $E_1$. For a
significantly higher value of $\gamma$, mites invade the environment, the disease becomes
endemic and no healthy bees survive, equilibrium $E_5$.

\begin{figure}[h!]
\centering 
{\includegraphics[scale=.7]{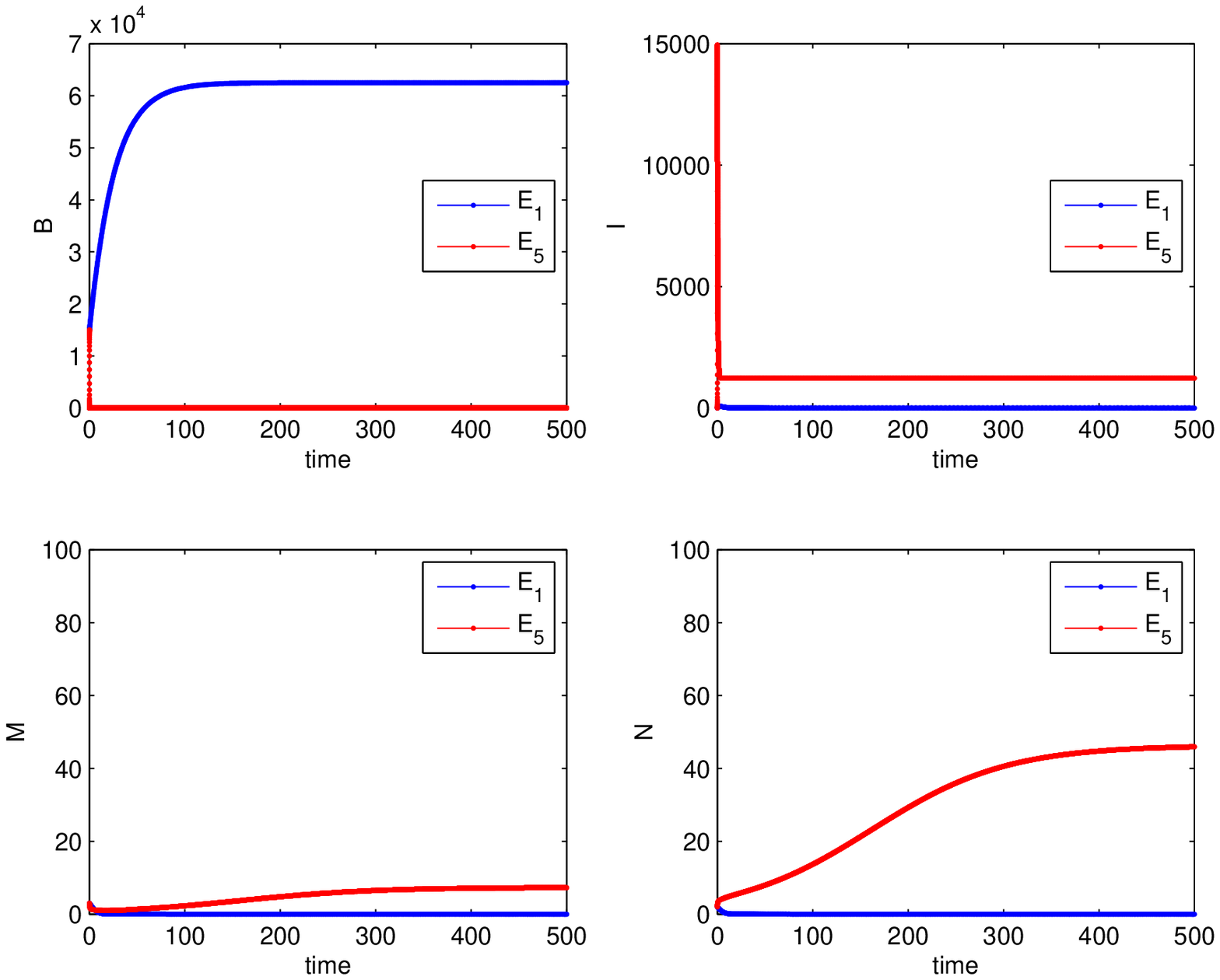}}
\caption{Transcritical bifurcation between $E_1$ and $E_5$, for the parameters values
$r=0.02$, $K=15000$, $b=2500$, $m=0.04$, $n=0.007$, $e=0.00001$,$\mu=2$, $\lambda=0.004$,
$\delta=0.0005$, $\beta=0.0001$, $\gamma=0.00002$ (blue), $\gamma=0.3$ (red), $p=0.0003$.
Initial conditions $B=15000$, $I=0$, $M=3$, $N=2$.}
\label{transcE0-E4}
\end{figure}

\item $E_5$-$E^*$

In Figure \eqref{transcE4-coes} we discover that a transition can occur also
between $E_5$ and $E^*$,
if the horizontal transmission rate of the virus among honey bees $\gamma$
is low enough. In such case indeed the healthy bees can thrive in the hive.

\begin{figure}[h!]
\centering 
{\includegraphics[scale=.7]{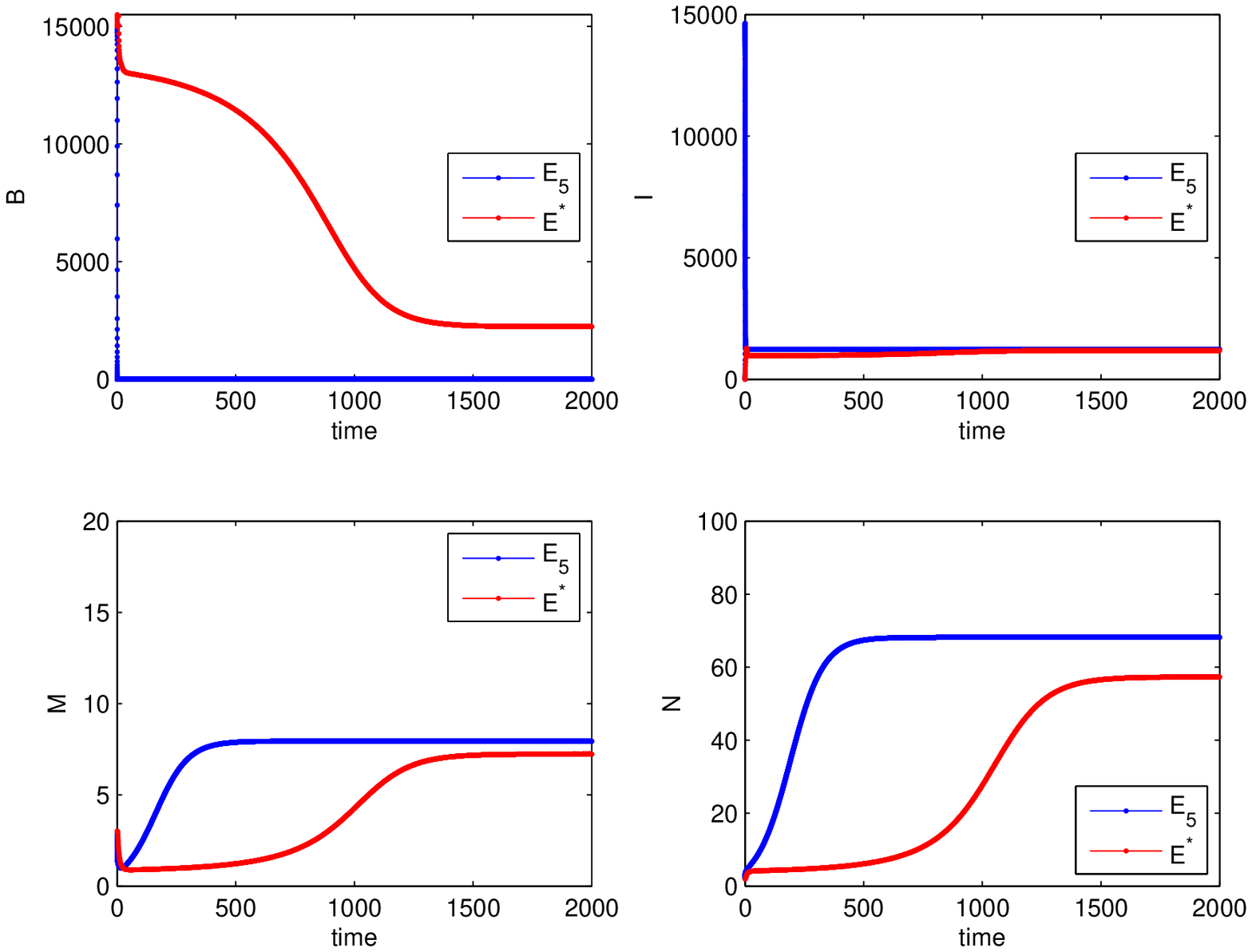}}
\caption{Transcritical bifurcation between $E_5$ and $E^*$, for the parameters values
$r=0.02$, $K=15000$, $b=2500$, $m=0.04$, $n=0.007$, $e=0.000001$, $\mu=2$, $\lambda=0.01$,
$\delta=0.001$, $\beta=0.0001$, $\gamma=0.03$ (blue), $\gamma=0.0001$ (red), $p=0.0002$.
Initial conditions $B=15000$, $I=0$, $M=3$, $N=2$.}
\label{transcE4-coes}
\end{figure}

\item $E_1$-$E^*$

Finally, in Figure \eqref{transcrE0-coes}, we can see that the system can change its
behavior going directly from $E_1$ to $E^*$, again by suitably tuning the parameter $\gamma$.

\begin{figure}[h!]
\centering 
{\includegraphics[scale=.7]{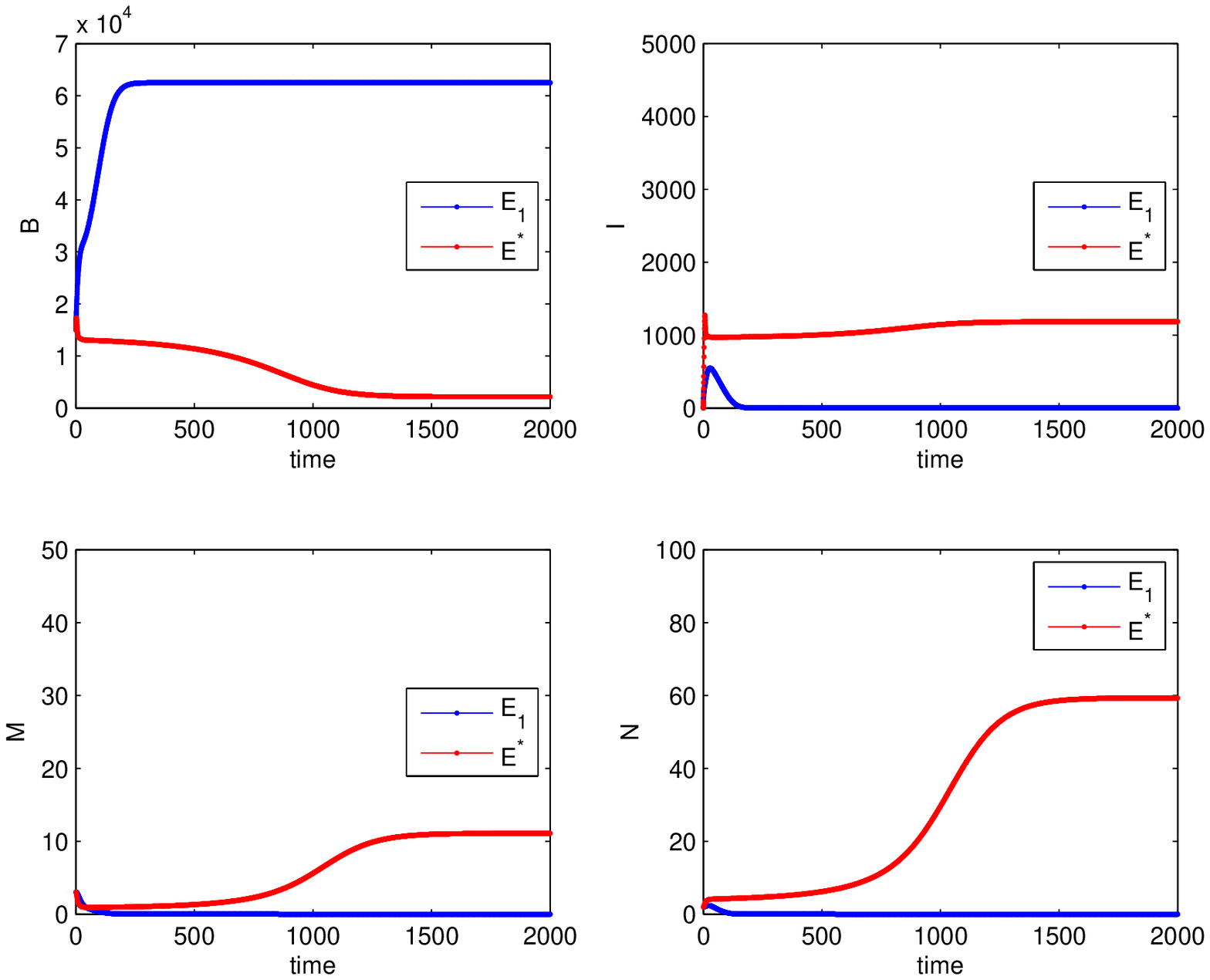}}
\caption{Transcritical bifurcation between $E_1$ and $E^*$, for the parameters values
$r=0.02$, $K=15000$, $b=2500$, $m=0.04$, $n=0.007$, $e=0.000001$, $\mu=2$, $\lambda=0.01$,
$\delta=0.0001$, $\beta=0.0001$, $\gamma=0.00002$ (blue), $\gamma=0.0001$ (red), $p=0.0002$.
Initial conditions $B=15000$, $I=0$, $M=3$, $N=2$.}
\label{transcrE0-coes}
\end{figure}

\end{itemize}

We also provide a bifurcation diagram for the four populations as a function of the
bifurcation parameter $\gamma$ in Figure \eqref{graficobello}. 
Starting from very low values of $\gamma$, we find at first the healthy-bees-only equilibrium
$E_1$. Note that only the healthy bee population $B$ thrives when
$\gamma <\gamma^{\dagger} \approx 5 \cdot 10^{-5}$.
Then in the range $\gamma^{\dagger} < \gamma < \gamma_1 \approx 1.4 \cdot 10^{-4}$
the disease becomes endemic, thus $I > 0$, while mites keep on being wiped out. The system is
found thus at equilibrium $E_4$. For values of $\gamma$ larger than $\gamma_1$ and up
to $\gamma_2 \approx 3 \cdot 10^{-3}$, the system quickly settles at the equilibrium $E^*$.
Here we have coexistence of bees and mites in an endemic hive, all the populations survive.
Past the threshold value $\gamma_2$
the healthy bee population vanishes. In this situation, the disease remains endemic
and affects all the bees.
This last transition shows numerically that we have another transcritical
bifurcation between the coexistence and the \textit{Varroa} invasion situations, namely
the equilibria $E^*$ and $E_5$. 

\begin{figure}[h!]
\centering 
{\includegraphics[scale=.7]{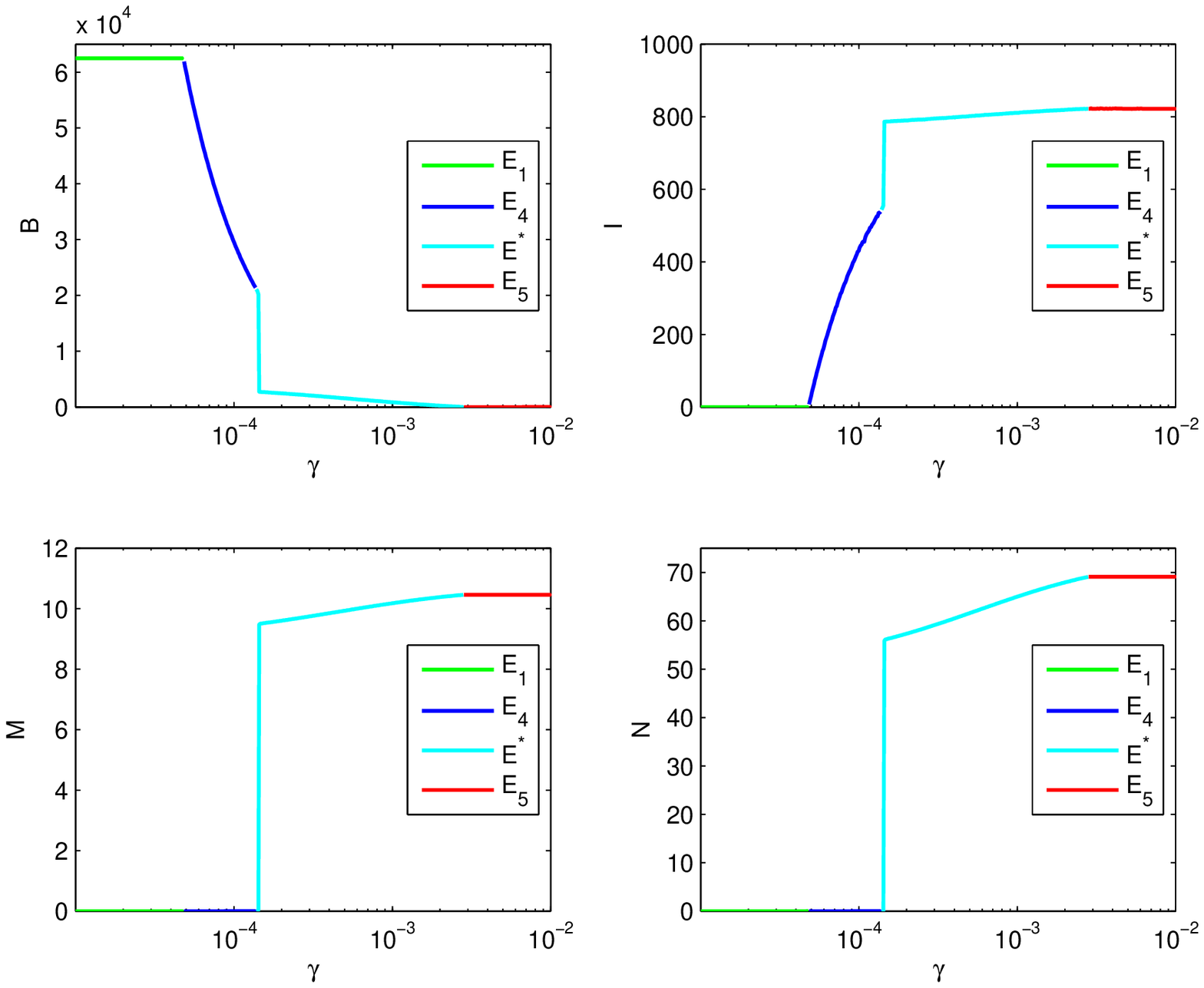}}
\caption{Bifurcation diagram, for the parameters values $r=0.02$, $K=15000$, $b=2500$,
$m=0.04$, $n=0.007$, $e=0.000001$, $\mu=3$, $\lambda=0.01$, $\delta=0.001$, $\beta=0.0001$,
$p=0.0002$. Initial conditions $B=15000$, $I=0$, $M=3$, $N=2$.}
\label{graficobello}
\end{figure}

\subsection{Hopf bifurcations}

We now try to establish whether there are special parameter combinations for which
sustained population oscillations are possible.
For this purpose, the eigenvalues must cross the imaginary axis.

This is easy to assess for a quadratic characteristic equation,
$\Lambda^2 + b\Lambda + c = 0$, since we need the linear term to vanish,
$b = 0$, and the constant term to be positive, $c > 0$.

Clearly at $E_1$ no Hopf bifurcation arises, since the eigenvalues are all real.

At $E_4$ the characteristic equation factors into the product of two quadratics.
The first one, namely the equation \eqref{first}, has a negative linear term so we exclude
the possibility for a Hopf bifurcation to occur. 

The second one has the form \eqref{second}.
To see whether Hopf bifurcation arise, we need the trace to vanish and the determinant
to be positive. Thus, imposing the condition on the trace we are led to
\begin{align}
&r-\frac{\beta(b\gamma-m\mu)}{\mu\gamma}+\frac{2e(b\gamma-m\mu-\mu^2)}{\mu \gamma}-n =0,
\label{1}
\end{align}
This condition however contradicts the second stability condition (\ref{E4stab}).
Indeed, note that solving for $n$ the equation \eqref{1}, we find
\begin{equation}
n=r-\frac{\beta(b\gamma-m\mu)}{\mu\gamma}+\frac{2e(b\gamma-m\mu-\mu^2)}{\mu \gamma},
\end{equation}
and the second inequality in \eqref{E4stab} becomes
$$
-\left(r-\frac{\beta(b\gamma-m\mu)+e(-b\gamma+m\mu+\mu^2)}
{\mu\gamma}\right)^2-\frac{r\beta(b\gamma-m\mu)}{\mu\gamma}>0,
$$
which is never satisfied, using the left inequality in the
feasibility condition \eqref{E4feas}.
We conclude that at $E_4$ no Hopf bifurcations can arise.

\section{Sensitivity analysis}\label{sec:sens}

In this section we perform the sensitivity analysis on \eqref{model} in order to rank the
parameters with respect to their influence on the system dynamics.

The first step is to determine the \textit{sensitivity equations} by differentiating the
original equations with respect to all the parameters. 
Our model consists of 4 equations with 12 parameters, therefore we get a sensitivity system
of 48 equations. 

In fact, \cite{comin}, when defining $\alpha$ the vector of the 12 parameters, the sensitivity system for
\eqref{sistbound} can be formulated as
\begin{equation}
\frac{d}{dt}\left( \frac{\partial X_i ( \alpha , t)}{\partial \alpha_j} \right)
= \sum_{s=1}^{4} \frac{\partial f_i ( X, \alpha,t) }{\partial X_s}
\frac{\partial X_s(\alpha,t)}{\partial \alpha_j}
+\frac{\partial f_i(X,\alpha, t)}{\partial \alpha_j}, 
\end{equation}   
for $i=1, \dots, 4$ and $j=1, \dots, 12$, with initial conditions
\begin{equation}
\frac{\partial X_i ( \alpha, 0)}{\partial \alpha_j}=0, \quad i=1, \dots, 4, \quad j=1, \dots, 12.
\end{equation}

Note that solving such system requires the knowledge of the state variables $X_i, i=1\dots,4$
and therefore it involves also the solution of the system \eqref{model}. Finally, we are led to solve
a system of $52$ differential equations. 
  
The rate of change of the state as a function of the change in the chosen parameter
is obtained from the value of the sensitivity solution evaluated at time $t$, \cite{BN}.

However, the parameters, and thus also the sensitivity solutions for different parameters, are
measured in different units. This means that without further action, any comparison is futile.
To make a valid comparison of the effect that parameters with different units have on the solution,
we simply multiply the sensitivity solution by the parameter under consideration. This form of the
sensitivity function is known as the \textit{semi-relative sensitivity solution}, \cite{BB}. It provides
information concerning the amount of change that the population will experience when the parameter
of interest is subject to a positive perturbation, \cite{BN}.

Computing the semi-relative solutions and making a comparison of the scales on the vertical axis of
these plots (formally, this is equivalent to ranking the parameters according to the $\infty -$norm),
we deduce that the parameters most affecting the system are $r, \mu $ and $\gamma$.

Specifically, an increase of the \textit{Varroa} growth rate has a negative effect on the bee
population and a positive influence on the mite population, consistently with Figure \eqref{sens_r}
(left frame). 
Furthermore, the effect of perturbing $r$ can be observed best on the population of healthy bees and
results in a reduction of $1200$ units (bees) after $210$ days.   

\begin{figure}[h!]
{\includegraphics[scale=.35]{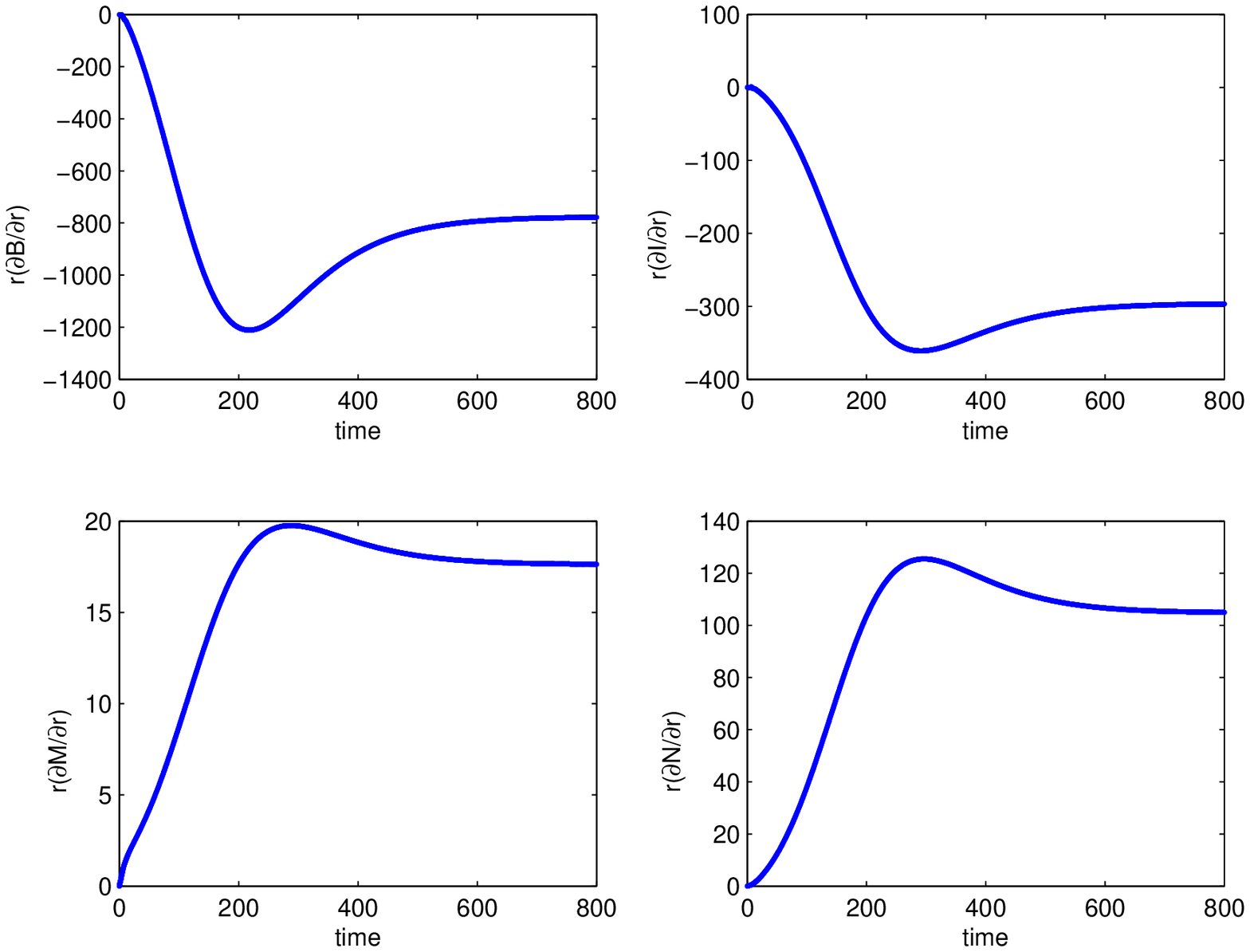} \hfill \includegraphics[scale=.35]{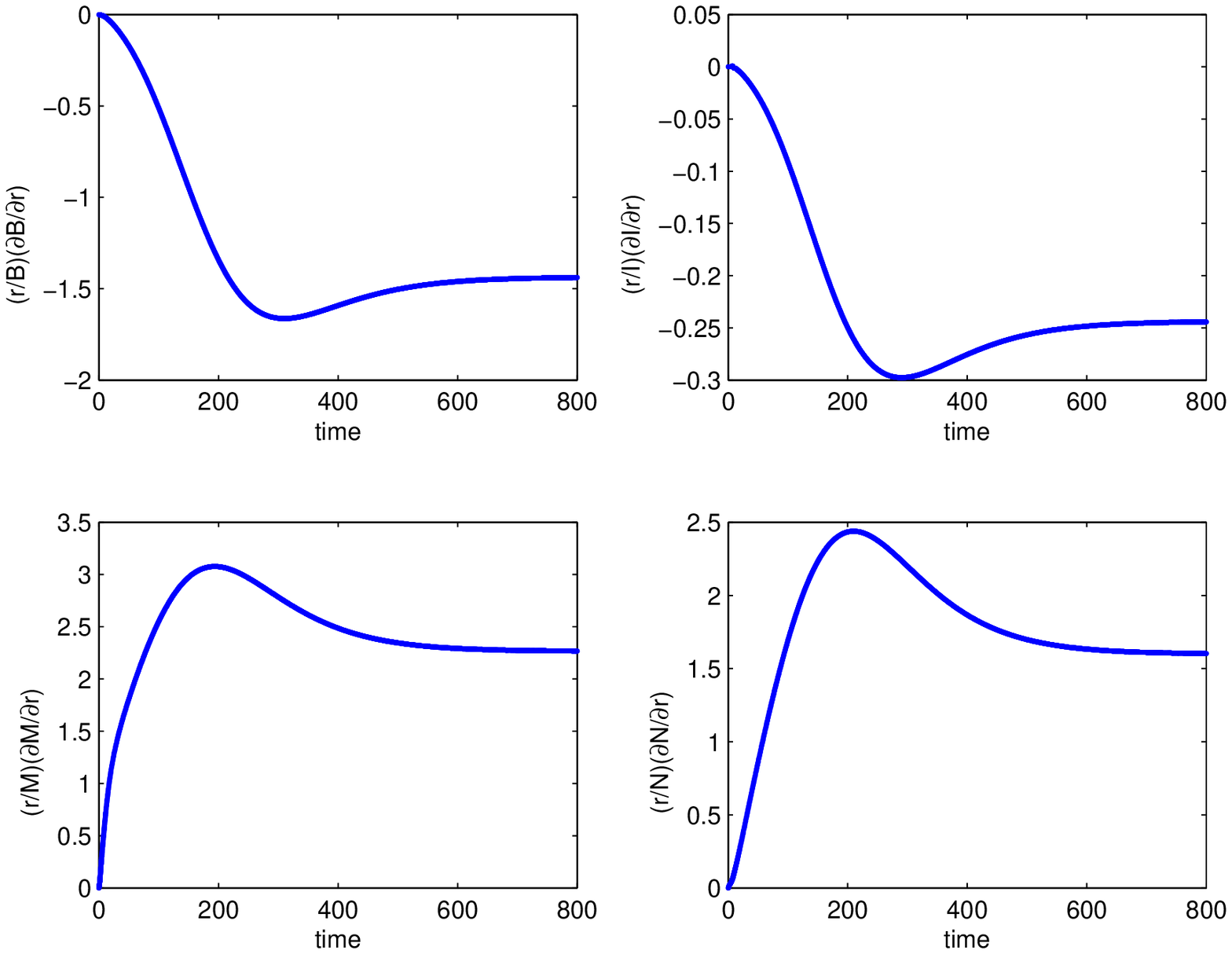}}
\caption{Semi-relative solutions (left) and the logarithmic sensitivity solutions (right) with
respect to the parameter $r$.}
\label{sens_r}
\end{figure}

From Figure \eqref{sens_r} (right frame), we can approximatively get the expected percentage
changes from an increase of $r$. In particular, \cite{BN}, we compute the logarithmic sensitivity solutions,
namely
\begin{equation}
\label{log_sens}
\frac{\partial \log (X_i (\alpha, t))}{\partial \log (\alpha_j)}=\frac{\alpha_j}{X_i(\alpha , t)}
\frac{\partial X_i(\alpha , t)}{\partial \alpha_j}, 
\end{equation} 
for $i=1, \dots, 4$ and $j=1, \dots, 12$.

Calculating \eqref{log_sens}, we find that a positive perturbation of $r$ leads to a decrease of
up to 170\% in the healthy bee population at $t=310$ days, a 30\% decrease in the infected bee
population at $t=290$ days, a 300\% increase in the healthy mite population at $t=200$ days, and
a 250\% increase in the infected mite population at $t=200$ days.

However, we also note that over 200\% changes, at time $200$ days, in the mite populations are not
significant due to their low level of concentration at time $200$ days in the solution. This is shown
in Figure \eqref{sim_coes2malata} for the set of parameter values used to perform sensitivity analysis.

We now turn to the parameter $\mu$, second in order of sensitivity. From Fig. \eqref{musens}, it is
clear that a higher value for the disease-related mortality of the bees has a positive impact on the
healthy bee population and a negative influence on the infected bee population.
In particular, at the end of the observation period we find an increase of $1000$ units (about 200\%)
in the healthy bee population  and a reduction of $600$ units (less than 50\%) in the infected bee population.

\begin{figure}[h!]
\centering
\includegraphics[scale=.5]{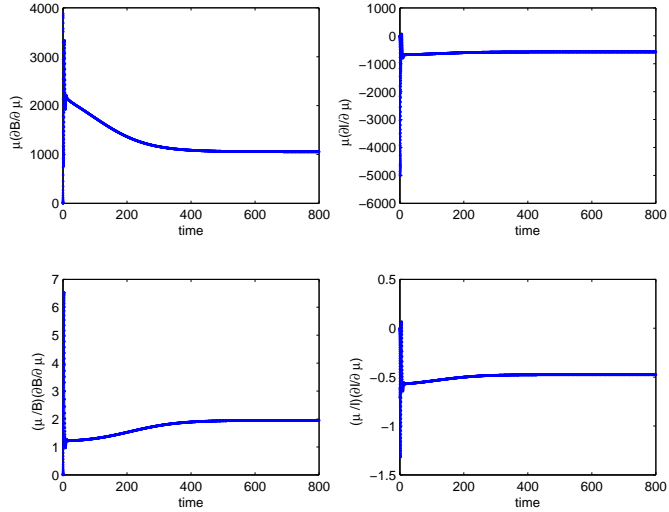} 
\caption{Semi-relative solutions (top) and logarithmic sensitivity solutions (bottom) for the populations
B (left) and I (right) with respect to the parameter $\mu$.}
\label{musens}
\end{figure}

According to our results, an increase in $\mu$ inevitably leads to a decrease in the infected bee
population. The shorter time of survival means that there is also less time for the virus to be
transmitted. This gives an explanation for the positive impact on the healthy bee population.
This is in line with our earlier observations: when transmitted by Varroa mites, the most virulent
diseases at the colony level are the least harmful ones to the single bees. 
 
Lastly we find the negative influence of $\gamma$ on the population of healthy bees with an
expected decrease of $600$ units and a percentage change of about 100\% at the end of the observation
period, as shown in Fig. \eqref{sens_gamma}.  

\begin{figure}[h!]
\centering
\includegraphics[scale=.35]{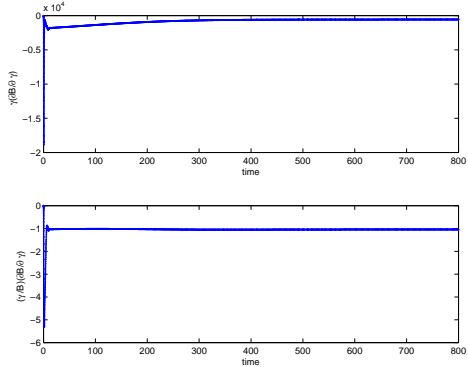} 
\caption{Semi-relative solution (top) and logarithmic sensitivity solution (bottom) for the population
B with respect to the parameter $\gamma$.}
\label{sens_gamma}
\end{figure}

To conclude the analysis, we look at Fig. \eqref{npbsens}.  Both a positive change in the \textit{Varroa}
natural mortality rate, $n$, as well as in their intraspecific competition coefficient, $p$, have a
positive effect on the healthy bee population and, unexpectedly, also on the infected bee population. 
This suggests that the horizontal transmission of the virus among honey bees would increase the
infected bee population also for higher values of $n$ and $p$, i.e. for a lower level of
\textit{Varroa} infestation.
The influence of the parameter $\gamma$ on the bee population also explains the frame on the
bottom: even a positive perturbation in the honey bee birth rate will yield a decrease in the
population of healthy bees, which soon leave this class to benefit the population of the infected ones. 
Furthermore, the latter is really sensitive with respect to $b$, with a dramatic increase of
about $1200$ units (100\%) as early as time $t=10$.

\begin{figure}[h!]
\centering
{\includegraphics[scale=.35]{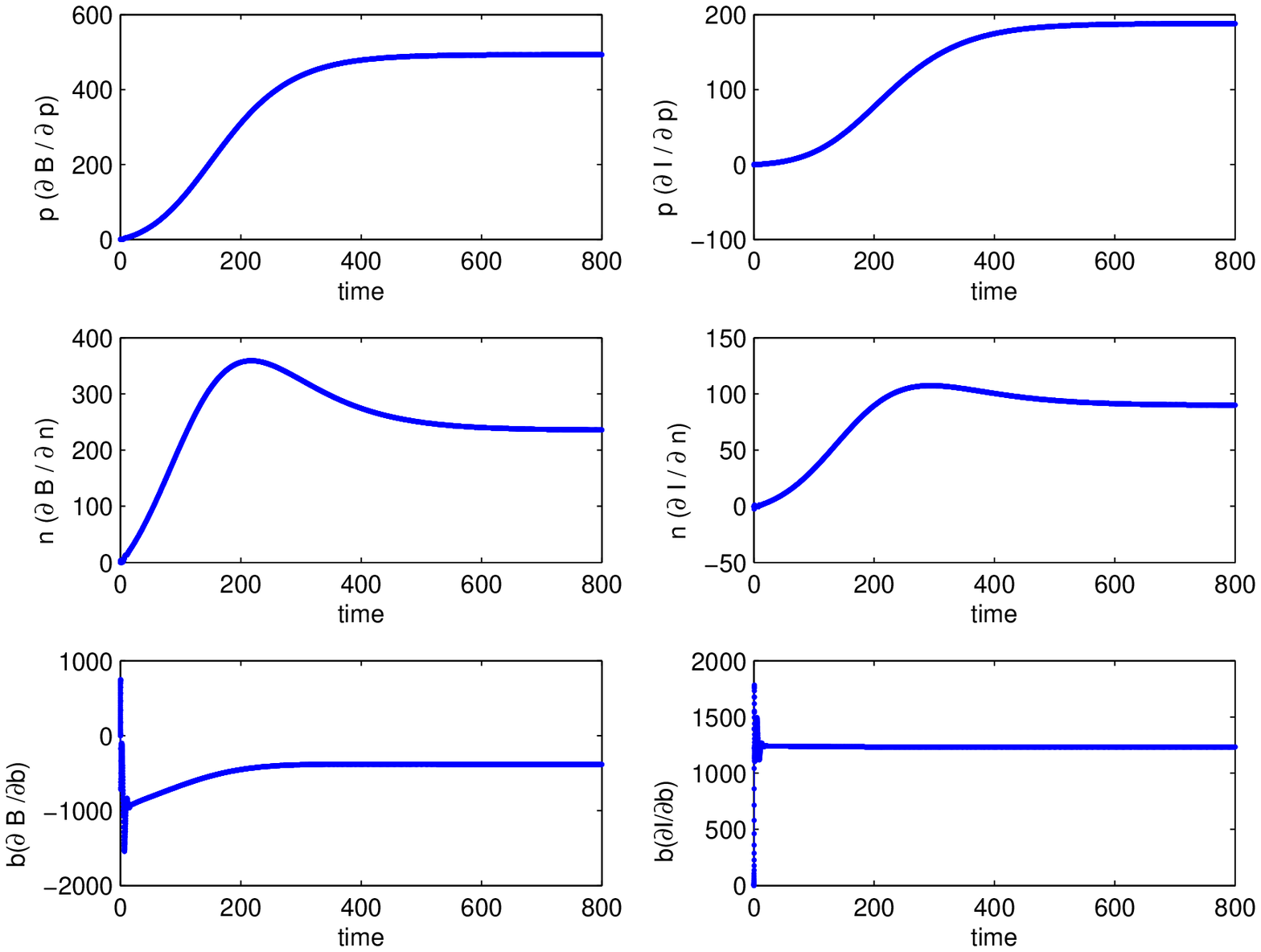}  \hfill \includegraphics[scale=.35]{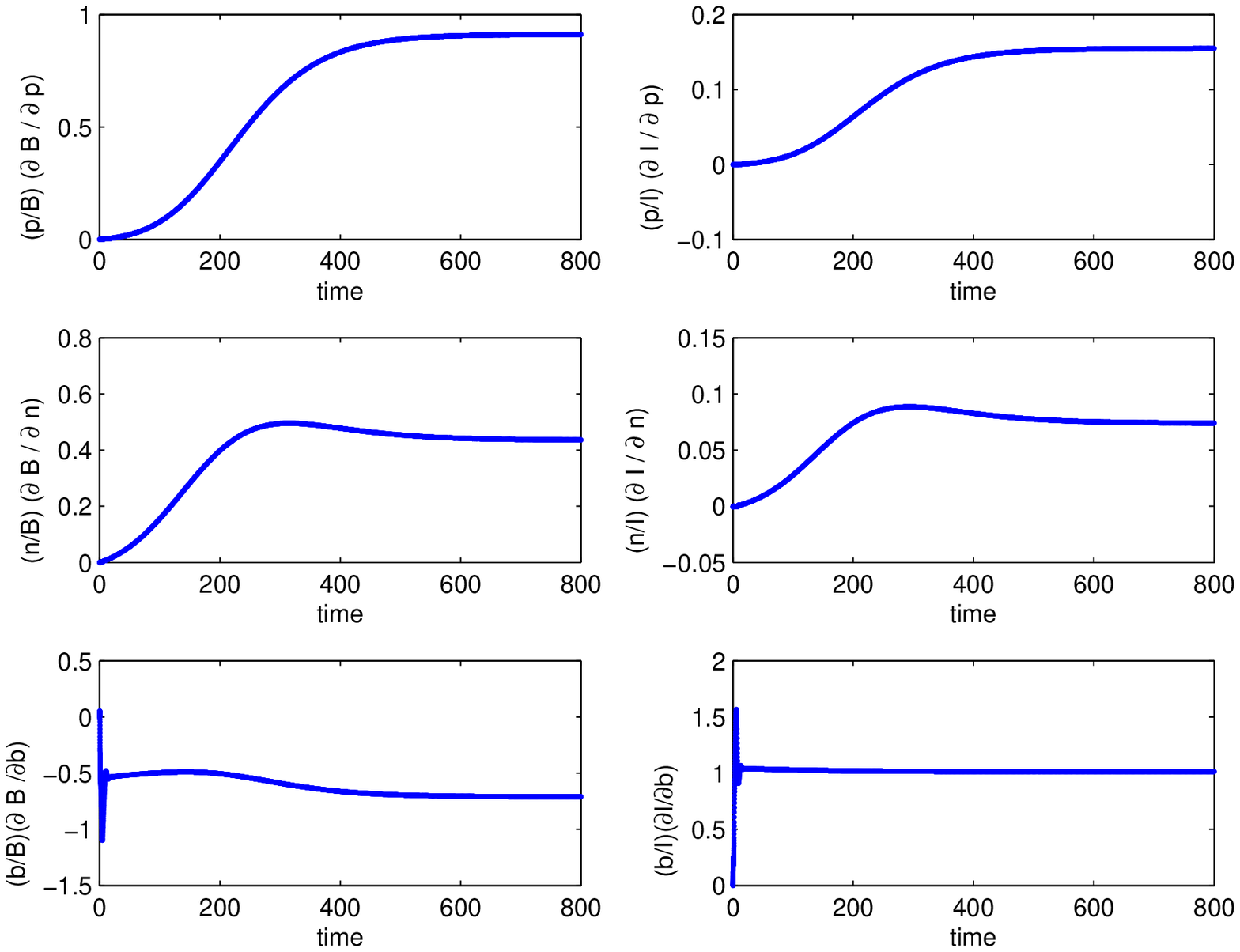}}
\caption{Semi-relative solutions (left) and logarithmic solutions (right) for the populations B and
I with respect to the parameters $p$, $n$ and $b$.}
\label{npbsens}
\end{figure}

\begin{oss}
The parameter of interest is usually perturbed by doubling or halving it, in order to understand how
sensitive the system dynamics are to such changes in it.
Therefore, for each parameter we perform the sensitivity analysis doubling and halving it to find out
that the effect of these perturbations is always the same and sensitivity solution plots differ only by
the scale on the temporal axis. In particular, doubling parameters yields to a shorter temporal scale,
while a longer one is obtained by halving them.      
\end{oss}

The knowledge of the most influential parameters allows us to understand on which one of them we could
act in order to drive the system to a safer situation.  
From a theoretical point of view, we can conclude that the population of healthy bees would benefit
from both a reduction in the \textit{Varroa} growth rate, as well as a reduction in the horizontal
transmission rate among bees.

\section{Conclusion}

Field and research evidence indicate that the long-term decline of managed honey beehives
in the western countries is associated with the presence of viruses in the
honey bee populations, \cite{GA}. While formerly the viruses produced only covert infections,
their combination with the invading parasite \textit{Varroa destructor}
has triggered the emergence of overt viral infections. These cause concern
because they entail fatal symptoms at both the individual bee level and at the colony
level, constituting a major global threat for apiculture.

In model presented here
we combine bee and mite population dynamics with viral epidemiology, to produce
results that match field observations well and give a clear explanation of how
\textit{Varroa} affects the epidemiology of certain naturally occurring bee viruses,
causing considerable damages to colonies.

In the model, there are only four possible stable equilibria, using the known field
parameters, see Table \ref{table2}. The first one contains only the thriving healthy bees.
Here the disease is not present and also the mites are wiped out. Alternatively,
there is another equilibrium where the mite population still disappears,
but the disease remains
endemic among the bee population that survives.
Thirdly, infected bees coexist with the mites in the \textit{Varroa} invasion scenario; in
this situation the disease invades the hive, affecting all the bees and driving the
healthy bees to extinction. The final coexistence equilibrium is also possible, with
both populations of bees and mites thriving and with an endemic
disease in the beehive among both species. Transcritical bifurcations relate these points
while, in spite of the four dimensionality of the system, which may render them more likely,
Hopf bifurcations have been shown not to arise.

Further, two alternatives turn out to be impossible by the use of field data for the
parameter values.

On the one hand, the endemic disease cannot affect all the bees
in a \textit{Varroa}-free colony.
According to the envisioned role of \textit{V. destructor} as disease vector
for increasing virulence of honey bee viruses, from the analysis
the disappearance of whole healthy bee population in a \textit{Varroa}-infested colony
seems much more likely than in a \textit{Varroa}-free colony.

On the other hand, the healthy mites cannot thrive only with healthy bees. Namely, if
the \textit{Varroa} population is present, then necessarily the bees viral infection occurs.

The findings of this study also indicate that a low horizontal virus transmission rate
among honey bees in beehives will help in protecting the bee colonies from the
\textit{Varroa} infestation and the viral epidemics. 

The model presented here gives good qualitative insight into the spread
of the bee disease. However, in its assumptions
certain effects that can be relevant are omitted.
Therefore, its primary value lies in the context of qualitative understanding
rather than that of a quantitative prediction.
For instance, the simulations use constant parameter values, but
it remains to be investigated whether the qualitative results would change
if continuously changing parameters instead were considered. 


Similarly, also the intraspecific competition coefficient among \textit{Varroa} $p$,
taken here as constant, could reasonably be considered as a bee-population-dependent
function, since mites essentially compete for the host hemolymph.

Nevertheless, despite this simplicity the model presented here might be a good starting
point for the development of other more sophisticated models including seasonality
and other aspects relevant for the bee colony health, or for example to
investigate the effectiveness of various \textit{Varroa} treatment strategies.

\end{document}